\documentclass[letterpaper, 10 pt, conference]{ieeeconf}  

\IEEEoverridecommandlockouts                              
\usepackage[noend]{algorithmic}
\usepackage{algorithm}

\usepackage{times}
\usepackage{amsmath,amsfonts,dsfont}
\usepackage{verbatim}
\usepackage{multicol}
\usepackage{amssymb,amsmath}
\usepackage{wrapfig}

\usepackage[bookmarks=true]{hyperref}
\usepackage{graphicx}

\newtheorem{thm}{\textbf{Theorem}}
\newtheorem{definition}{\textbf{Definition}}
\newtheorem{assumption}{Assumption}
\newtheorem{corollary}{\textbf{Corollary}}

\DeclareFontFamily{U}{mathx}{}
\DeclareFontShape{U}{mathx}{m}{n}{<-> mathx10}{}
\DeclareSymbolFont{mathx}{U}{mathx}{m}{n}
\DeclareMathAccent{\widehat}{0}{mathx}{"70}
\DeclareMathAccent{\widecheck}{0}{mathx}{"71}

\title{\LARGE \bf
Provably Stable Multi-Agent Routing with Bounded-Delay \\
Adversaries in the Decision Loop}

\author{Roee M. Francos*, Daniel Garces*, and Stephanie Gil
\thanks{(*Co-primary authors) R.M.~Francos, D.~Garces and S.~Gil are with the School of Engineering and Applied Sciences, Harvard University, Cambridge, MA 02138 USA (e-mails: {\tt\small rfrancos@seas.harvard.edu, dgarces@g.harvard.edu, sgil@seas.harvard.edu}).}
}

\begin{document}

\maketitle
\thispagestyle{empty}
\pagestyle{empty}

\begin{abstract}
 In this work, we are interested in studying multi-agent routing settings, where adversarial agents are part of the assignment and decision loop, degrading the performance of the fleet by incurring bounded delays while servicing pickup-and-delivery requests.  Specifically, we are interested in characterizing conditions on the fleet size and the proportion of adversarial agents for which a routing policy remains stable, where stability for a routing policy is achieved if the number of outstanding requests is uniformly bounded over time. To obtain this characterization, we first establish a threshold on the proportion of adversarial agents above which previously stable routing policies for fully cooperative fleets are provably unstable. We then derive a sufficient condition on the fleet size to recover stability given a maximum proportion of adversarial agents. We empirically validate our theoretical results on a case study on autonomous taxi routing, where we consider transportation requests from real San Francisco taxicab data.
\end{abstract}

\section{Introduction}

In this paper we focus on a routing setting where a fleet of agents must pick up and deliver stochastically appearing requests. This stochastic setup is common in mobility-on-demand \cite{Alonso2017, garces2023multiagent, garces2024approximate} and warehouse logistics \cite{Chen2021, Wilde2024}, where the location and quantity of future requests are unknown in advance. We assume that each agent handles one request at a time. In our setup, a subset of agents in the fleet may act adversarially by deviating from the prescribed plan set by the centralized control system, resulting in longer than expected service times for their assigned requests. This service delay model is inspired by operations research studies \cite{mari2025online}, particularly in transportation and delivery systems \cite{smith2008dynamic,zhang2016control}, where drivers, after accepting a request, may pause for personal breaks or take longer routes to increase earnings when compensated per mile. We assume that if the agents take too long to service a request, then the system will remove them, hence agents can only incur a bounded delay. Hereafter we refer to this as the \emph{bounded-delay model} for adversaries. Our objective in this paper is then to characterize conditions on the fleet size and the proportion of adversarial agents in the system for which a routing policy is provably stable in the presence of bounded delay adversarial agents, where a stable routing policy is one that guarantees that the number of outstanding requests is uniformly bounded over time.

Previous work in the literature has studied necessary and sufficient conditions for stability of fully cooperative fleets of agents \cite{spieser2014, zhang2016control, garces2024approximate}. Guaranteeing stability ensures that outstanding requests are bounded uniformly over time, which is often desirable in an effective routing policy.
However, the assumption that a fleet acts fully cooperatively is a strong assumption, often not holding in practical settings, where agents can deviate from their assigned plan \cite{yu2022delay}. This deviation can cause increased wait times and commute times, which might result in requests not being serviced and accumulating over time, leading to instability of a routing policy.

To address the possible influence of adversaries in multi-agent routing problems, related works in the literature primarily focus on the impact of adversarial agents external to the controlled fleet \cite{Yuan2016, Thai2018, Kearney2018induction, shishika2018local, shishika2019team, kailkhura2014asymptotic, Altman2010adversarial, Zhu2012deceptive, Chu2023, Chu2024, Sharma2019attacks, Qayyum2020securing, Guo2023, Poudel2024, cavorsi2022adaptive, yu2024sensing, francos2021search}. These studies typically focus on scenarios where adversaries reduce the number of agents available to service requests, thereby impairing overall fleet performance. However, these works do not consider settings in which adversarial agents remain within the system and cannot be removed. In such cases, compromised agents continue to participate in the decision-making loop, receive task assignments, and consequently degrade system performance from within. Addressing this vulnerability is crucial for mitigating the impact of both intentional and unintentional adversaries on the system. Ideally, stability can be restored and maintained even when a subset of the agents in the network is adversarial.

In this paper we characterize conditions under which previously known results on stability break, and we derive new conditions to recover stability in the presence of bounded-delay adversaries.
Specifically, we are interested in characterizing conditions on the proportion of adversarial agents and the size of the agent fleet under which routing policies are stable. To achieve this, we derive an upper bound on the adversaries' impact by quantifying their wasted effort within the fleet. This inefficiency arises from the indirect routes adversaries take while fulfilling service requests. 
For this work, we consider an instantaneous assignment routing policy, which solves a bipartite matching problem using incoming requests and available agents. We select this policy due to its widespread adoption in applications like mobility-on-demand \cite{Alonso2017, garces2023multiagent, Rivas2019} and the existence of theoretical bounds on the fleet size for fully cooperative (non-adversarial) fleets \cite{garces2024approximate}, which allows us to readily characterize the impact of adversarial agents on the system.

Our contributions in this paper are as follows. First, we characterize as a function of known problem parameters, such as the delay introduced by the adversarial agents, the expected number of requests, and the map size, two main results:
\begin{itemize}
    \item A threshold on the proportion of adversarial agents beyond which previously considered routing policies \cite{garces2024approximate} are provably unstable.
    \item A sufficient condition on the fleet size required for stability of the instantaneous assignment policy under the presence of bounded-delay adversarial agents. 
    \end{itemize} 
Lastly,
    \begin{itemize}
    \item We consider a case study on urban transportation to validate our theoretical results. We first empirically demonstrate that existing bounds on the fleet size required for stability of previously studied routing policies — such as random and instantaneous assignment — fail in the presence of adversarial agents. We then empirically show that our newly derived sufficient condition successfully recovers stability. For our simulations, we consider real transportation requests for San Francisco taxicabs ~\cite{piorkowski2009crawdad}. Notably, our empirical results highlight the fact that stability can be restored even when the number of added cooperative agents is less than the number of adversaries.
\end{itemize}

Our theoretical and empirical results provide guidance for system designers in evaluating whether a given fleet size can tolerate a specified proportion of adversarial agents while maintaining stability. When the adversarial proportion exceeds the tolerable threshold, our derived fleet size condition can be used to determine an appropriate fleet size that guarantees stability under a worst-case level of adversarial presence.

\section{Problem Formulation}
In this section, we formulate the problem of routing a fleet of agents to fulfill on-demand pickup-and-delivery requests when the fleet is composed of both cooperative and adversarial agents. Our goal is to characterize conditions on the fleet size that ensure the stability of a routing policy, even when adversarial agents are present. In the following subsections we provide formal definitions for the environment and control space, the request model, the adversarial influence model, and our definition of stability.

\subsection{Environment and Control Space}
We assume that we have a centralized server that has access to the location of all agents. We assume that we have discrete time steps and the system runs for a fixed time horizon of length $T$.
We represent the environment for the routing problem as a directed graph with a fixed topology. We denote the directed graph as  ${\mathcal{G}} = \left( {\mathds{V},\mathds{E}} \right)$ where $\mathds{V}$ corresponds to the set of nodes in the graph, while $\mathds{E} \subseteq \{ (i,j) | i, j \in \mathds{V}\} $ corresponds to the graph's edges. 
We denote the set of neighboring nodes to node $i$ as $\mathcal{N}_i$, where ${\mathcal{N}_i} = \left\{ {j|j \in \mathds{V},(i,j) \in \mathds{E}} \right\}$.

We assume the fleet of agents is composed of both adversarial agents and cooperative agents.
We denote the set of adversarial agents as $\mathcal{A}$ and the set of cooperative agents as $\mathcal{C}$. Sets $\mathcal{A}$ and $\mathcal{C}$ are unknown and set membership for agents can change over time. We denote the full fleet size as $N$, where $N=|\mathcal{A}|+|\mathcal{C}|$. The fleet size $N$ is assumed to be fixed for the entire time horizon $T$. We denote the proportion of adversarial agents in the system as $F = \frac{|\mathcal{A}|}{N}$.

We define the state of the system at time $t$ as ${x_t} = \left[{\overrightarrow {{\nu _t}} ,\overrightarrow {{\tau _t}} } \right]$, where $\overrightarrow {{\nu _t}}  = \left[ {\nu _t^1,...,\nu _t^N} \right]$ corresponds to the list of locations for all $N$ agents at time $t$, and $\overrightarrow {{\tau _t}}  = \left[ {\tau _t^1,...,\tau _t^N} \right]$ corresponds to the list of expected times remaining in the current trip for all $N$ agents.
If agent $\ell$ is available, then $\tau_t^\ell = 0$. Otherwise, $\tau_t^\ell \in \mathbb{N}^+$.

For simplicity, we assume that each edge of the network graph ${\mathcal{G}}$ can be traversed by any agent in one time step. 
For this reason, we define the control space for agent $\ell$ at time $t$ as $\textbf{U}_t^\ell \left( {{x_t}} \right) = {{\mathcal{N}}_{\nu _t^\ell }} \bigcup \left\{ \nu _t^\ell ,{\psi _r} \right\}$, where ${{\mathcal{N}}_{\nu _t^\ell }}$ corresponds to the set of adjacent nodes to $\nu_t^\ell$ the current location of agent $\ell$, and $\psi _r$ represents a pickup control that becomes available if there is a request available at the agent's location. The control space for the entire fleet is then expressed as the cartesian product ${\textbf{U}_t}\left( {{x_t}} \right) = \textbf{U}_t^1\left( {{x_t}} \right) \times  \cdot  \cdot  \cdot  \times \textbf{U}_t^N\left( {{x_t}} \right)$.

\subsection{Request Model}
Following the notation for requests defined in \cite{garces2024approximate}, we define a pickup-and-delivery request as a tuple $ r = \left\langle {{\rho _r},{\delta _r},{t_r},{\phi _r}} \right\rangle$
where $\rho _r ,\delta_r  \in \mathds{V}$, are the request's desired pickup and drop-off locations, respectively. $t_r$ denotes the time at which the request entered the system, and $\phi _r$ is an indicator function corresponding to a binary pickup status of the request. Requests appear stochastically.

We model the request distribution using three random variables: 1) $\eta$ for the number of requests entering the system at each time step; 2) $\rho$ for the pickup location of a request; and 3) $\delta$ for the drop-off location of a request. We assume that the pickup and drop-off locations of different requests are independent and identically distributed.

\subsection{Adversarial Influence Model}
In this study, adversarial agents deviate from the prescribed plan set by the centralized control system. Specifically, we adopt a bounded-delay model similar to the one introduced in \cite{mari2025online}, where agents may incur a bounded delay while servicing requests. To formally define this model, we define $\Delta$ as the maximum allowable delay an agent can incur during a request's pickup or drop-off. Additionally, we assume that adversarial agents remain stationary at their current location until they are assigned a request. We denote the actual delay incurred by an adversarial agent $\ell_a$ while picking up a request $r$ as the error term $e^{\text{pick}}(\ell_a, r)$. Similarly, we denote the actual delay incurred by agent $\ell_a$ while dropping off request $r$ as the error term $e^{\text{drop}}(\ell_a, r)$. Using these delays, we define the bounded-delay model as follows:

\begin{definition}
    An adversarial agent $\ell_a \in \mathcal{A}$ located at node $v_a \in \mathds{V}$ is following the bounded-delay model if the following two conditions are satisfied:
    \begin{enumerate}
        \item Agent $\ell_a$ remains stationary at location $v_a$ until a request gets assigned to it.
        \item Once a request $r$ gets assigned to agent $\ell_a$, it can incur a bounded delay such that:
        \begin{align*}
            d_{\text{adv}}(v_a, \rho_r) & = d(v_a, \rho_r) + e^{\text{pick}}(\ell_a, r) \leq d(v_a, \rho_r) + \Delta \\
            d_{\text{adv}}(\rho_r, \delta_r) & = d(\rho_r, \delta_r) + e^{\text{drop}}(\ell_a, r) \leq d(\rho_r, \delta_r) + \Delta
        \end{align*}
        where $d_{\text{adv}}(v_a, \rho_r)$ corresponds to the time it takes the adversarial agent to travel from its current location to the pickup location of request $r$, and $d_{\text{adv}}(\rho_r, \delta_r)$ corresponds to the time it takes the adversarial agent to travel from the pickup location to the drop-off location of request $r$.
    \end{enumerate}
    \label{def:adversarial_model}
\end{definition}

Under this delay model, adversarial agents may remain occupied for longer, which prevents them from servicing other requests. This reduction in availability of agents results in a potential accumulation of requests. This phenomenon is the root cause for the instability of routing policies that do not account for the presence of adversarial agents. 

\subsection{Stability of a Policy}
We define a policy $\pi  = \left\{ {{\mu _1},{\mu _2},...} \right\}$ as a set of functions that map state $x_t$ into control ${u_t} = {\mu _t}\left( {{x_t}} \right) \in {\textbf{U}_t}\left( {{x_t}} \right)$. 
Generalizing the formulations in \cite{spieser2014} and \cite{garces2024approximate}, we define the total time traveled in service of request $r_q$ when following a policy $\pi$ as $W_{r_q}^{\pi} = d(v_{r_q}^{\pi}, \rho_{r_q}) + d(\rho_{r_q}, \delta_{r_q})$, where $v_{r_q}^{\pi}$ corresponds to the location of the agent assigned to request $r_q$ according to policy $\pi$, and $\rho_{r_q}$ and $\delta_{r_q}$ correspond to the pickup and drop-off locations for request $r_q$, respectively. We assume that the selected policy $\pi$ moves agents closer to their assigned request once an assignment has been executed. We define $Z_{T}^{\pi}$ as the total time needed for policy $\pi$ to service all requests that enter the system during a horizon of length $T$, assuming that we have at least as many available agents as incoming requests at each time step. The expression for $Z_{T}^{\pi}$ is given by:
\begin{equation}
    Z_{T}^{\pi}=\sum_{t=1}^T \sum_{q=R_{t-1}+1}^{R_{t}} W_{r_q}^{\pi}
\end{equation} 
Where $R_t = \sum\limits_{t' = 1}^t {{\eta _{t'}}}$ is the total number of requests that have entered the system until time $t$ given that $\eta_{t'}$ is the instantiation of the random variable $\eta$ at time $t'$. 

We define the effective amount of time that an arbitrary fleet of agents can spend servicing requests given a policy $\pi$ during a horizon of length $T$ as $Q_{T}^{\pi}$. For the case where the fleet is fully cooperative and composed of $N$ agents, $Q_{T}^{\pi} = NT$, since we have $N$ agents that each move one time step for a time horizon of length $T$. In the case of adversarial agents, $Q_{T}^{\pi} \leq NT$ given that the adversarial agents incur a bounded delay and hence have less time available to service requests. In this paper, we generalize the arguments in \cite{spieser2014} and \cite{garces2023multiagent} to define stability using $Q_{T}^{\pi}$ as follows:
\begin{definition}
    A policy $\pi$ is stable if $E[Z_{T}^{\pi}] \leq Q_{T}^{\pi}$ for a horizon $T$.
    \label{def:routing_stability}
\end{definition}
This stability definition implies that if the expected amount of time to service requests exceeds the amount of time that agents can spend servicing requests then some requests can't be serviced by the current fleet size of $N$ agents under policy $\pi$. This situation then leads to accumulation of requests over time, which results in the number of outstanding requests growing unboundedly as $T \to \infty$. 

In this paper, we address the problem of characterizing a sufficiently large fleet size $N$ given the maximum allowable proportion of adversarial agents $F_{\text{max}}$ for which an instantaneous assignment policy $\bar{\pi}$ is stable according to Def.~\ref{def:routing_stability}, even in the presence of adversarial agents that follow the bounded-delay model defined in Def.~\ref{def:adversarial_model}. 

\section{Policy of Interest}
\label{sec:policies_def} 

In this work, we focus on an instantaneous assignment policy, which we denote as $\bar{\pi}$. The policy $\bar{\pi}$ assigns available agents to incoming requests by solving a bipartite matching problem that minimizes the time it takes the agents to service the requests. The bipartite matching solution is given by a matching algorithm, like the auction algorithm \cite{bertsekas1988auction, bertsekas2024new} or the modified JVC algorithm \cite{crouse2016implementing}. We assume that agents do not move unless a request has been assigned to them. In addition, once a request has been assigned to an agent, the agent cannot be assigned to a different request until it has serviced the originally assigned request. If a request $r_q$ is assigned to an agent $\ell_k$ and the request has not been picked up yet (i.e. $\phi_{r_q} = 0$), then at each time step, the agent will move towards $\rho_{r_q}$ the pickup location of request $r_q$ following the shortest path between the agent's original location and $\rho_{r_q}$. Once the agent has picked up the request (i.e. $\phi_{r_q} = 1$), the agent will move towards the drop-off location of request $r_q$ by following the shortest path between the pickup location $\rho_{r_q}$ and the drop-off location $\delta_{r_q}$.

We choose to focus on the instantaneous assignment policy since this policy is commonly used in routing applications \cite{Alonso2017, garces2023multiagent, garces2024approximate, garces2024surge}, it has previous theoretical bounds on the fleet size for fully cooperative fleets of agents \cite{garces2024approximate}, and its structured assignment makes it a reasonable choice for analysis when considering adversarial agents.

As a stepping stone in our analysis, we also consider a simple routing policy referred to as random assignment $\hat{\pi}$. This policy also moves agents that have been assigned to requests along the shortest path just like instantaneous assignment, but instead of solving a bipartite matching problem to generate the assignment, it assigns available agents to incoming requests uniformly at random, assigning one agent per request. The randomness of the assignment process for $\hat{\pi}$ makes the locations of agents and requests' pickups independent, simplifying the analysis. Once we obtain an upper bound for the amount of time required to service all requests under the random assignment policy, we prove that instantaneous assignment results in lower expected times needed to service requests, and hence the upper bound for random assignment also applies for the instantaneous assignment policy. Therefore, stability of the random assignment routing policy implies stability of the instantaneous assignment routing policy.

Now that we have the policies' definitions, we will present some results that we build on top of.

\begin{algorithm}
\caption{Random Assignment Policy}
\begin{algorithmic}[1]
    \label{alg:random_assignment_algo}
   \REQUIRE Current state ${x_t} = \left[ {\overrightarrow {{\nu _t}} ,\overrightarrow {{\tau _t}}} \right]$, and ${{\overline {\bf{r}} }_t}$ the set of outstanding requests at time $t$
   \ENSURE The agent control vector $\overrightarrow U$
   \STATE $\overrightarrow U = []$, $ \overrightarrow a = [1,...,N]$
   \FOR{$\ell \in [1,...,N]$}
           \IF{$\tau _t^\ell!=0$}
                \STATE h = \textbf{GetNextHopInGraph}$(\nu _t^{\ell}, \rho_r,\delta_r)$
                \STATE $\overrightarrow U[\ell]=h$
                \STATE $\overrightarrow a.\text{remove}(\ell)$
           \ENDIF
    \ENDFOR
    \FOR{$q \in  {{\bar{\textbf{r}}}_t}$}
    \STATE $\ell_q$ = \textbf{RandomlyChooseAvailableAgent}$(\overrightarrow a)$
            \STATE $ \overrightarrow a.\text{remove}(\ell_q)$
            \STATE {$\overrightarrow \tau _t^{\ell_q}=d(\nu _t^{\ell_q},\rho_{r_q})+d(\rho_{r_q},\delta_{r_q})$}
            \STATE h = \textbf{GetNextHopInGraph}$(\nu _t^{\ell_q}, \rho_{r_q},\delta_{r_q})$
             \STATE $\overrightarrow U[\ell_q]=h$
         \STATE {$\overrightarrow{{\tau _t}}[\ell_q]= \tau _t^{\ell_q} $}
    \ENDFOR
     \FOR{$\ell \in a$}
     \STATE $\overrightarrow U[\ell]=\nu _t^{\ell}$
     \ENDFOR
   \RETURN $\overrightarrow U$
\end{algorithmic}
\end{algorithm}

\begin{algorithm}
\caption{Instantaneous Assignment Policy}
\begin{algorithmic}[1]
\label{alg:instantaneous_assignment_algo}
   \REQUIRE Current state ${x_t} = \left[ {\overrightarrow {{\nu _t}} ,\overrightarrow {{\tau _t}}} \right]$, and ${{\overline {\bf{r}} }_t}$ the set of outstanding requests at time $t$
   \ENSURE The agent control vector $\overrightarrow U$
   \STATE $\overrightarrow U = []$, $ \overrightarrow a = [1,...,N]$, $\overrightarrow {{\nu _{t_q}}}  = \left[ {\nu _t^1,...,\nu _t^N} \right]$
   \FOR{$\ell \in [1,...,N]$}
           \IF{$\tau _t^\ell!=0$}
                \STATE h = \textbf{GetNextHopInGraph}$(\nu _t^{\ell}, \rho_r,\delta_r)$
                \STATE $\overrightarrow U[\ell]=h$
                \STATE $\overrightarrow a.\text{remove}(\ell)$
                \STATE $\overrightarrow {{\nu _{t_q}}}.\text{remove}(\nu _t^\ell)$
           \ENDIF
    \ENDFOR
     \STATE $\ell_{IA},\overrightarrow{r_{t_q}}$ = \textbf{PerformMatchingUsingAuctionAlg}$(\overrightarrow {{\nu _{t_q}}},{{\bar{\textbf{r}}}_t})$
    \FOR{$\ell \in  \ell_{IA}$}
            \STATE $ \overrightarrow a.\text{remove}(\ell)$
            \STATE $\rho_{r_q},\delta_{r_q}=\overrightarrow{r_{t_q}}[\ell]$
            \STATE {$\overrightarrow \tau _t^{\ell}=d(\nu _t^{\ell},\rho_{r_q})+d(\rho_{r_q},\delta_{r_q})$}
            \STATE h = \textbf{GetNextHopInGraph}$(\nu _t^{\ell}, \rho_{r_q},\delta_{r_q})$
             \STATE $\overrightarrow U[\ell]=h$
         \STATE {$\overrightarrow{{\tau _t}}[\ell]= \tau _t^{\ell} $}
    \ENDFOR
     \FOR{$\ell \in a$}
     \STATE $\overrightarrow U[\ell]=\nu _t^{\ell}$
     \ENDFOR
   \RETURN $\overrightarrow U$
\end{algorithmic}
\end{algorithm}

\section{Background: Necessary and Sufficient Conditions for Stability of a Policy with Only Cooperative Agents}
\label{sec:background}

Previous work in the literature \cite{garces2024approximate} has looked at necessary and sufficient conditions on the fleet size when the entire fleet of agents is cooperative (i.e. $|\mathcal{A}|=0$). For completion, we restate the two theoretical results that we build on top off.

We denote the initial location of an arbitrary agent $\ell$ at time $t=0$ by the random variable $\xi_\ell$. All initial locations $\xi_\ell$ for $\ell= 1,...,N$ are assumed to be i.i.d with a known underlying distribution $p_\xi$. 

For both, these background results and the results presented in this paper, we assume that $\eta$, $\rho$, and $\delta$ have underlying probability distributions $p_{\eta}$, $p_{\rho}$, and $p_{\delta}$, respectively, and these probability distributions can be estimated from historical data.

We denote $\pi_{\text{base}}$ as the routing policy of instantaneous assignment with reassignment, formally defined in \cite{garces2024approximate}.

\begin{thm}
    (Lemma 1 in \cite{garces2024approximate}) Let the random variable $v_{\text{rand}}$ with support $\mathds{V}$ represent the location of a random agent that gets assigned to a request after that agent has previously served a different request. Define $D_{\text{max}} \triangleq \max\{E[d(\xi, \rho)], E[d(v_{\text{rand}}, \rho)]\} + E[d(\rho, \delta)]$ to be the maximum expected time an agent needs to travel to service a request in the map. If the fleet size $N$ satisfies 
    $N \geq  E[\eta] \cdot D_{\text{max}}$
    then the random assignment policy $\hat{\pi}$ is a stable policy according to Def.~\ref{def:routing_stability}.
    \label{theorem:coop_sufficient_condition}
\end{thm}

\begin{thm}
    (Theorem 2 in \cite{garces2024approximate}) Let $\mathit{WD}(p_\delta, p_\rho)$ denote the first Wasserstein distance \cite{ruschendorf1985wasserstein} between probability distributions $p_{\delta}$ and $p_{\rho}$ with support $\mathds{V}$, such that:
    \begin{align*}
        \mathit{WD}(p_\delta, p_\rho) = \inf_{\gamma \in \Gamma(p_\delta, p_\rho)} \int_{x,y\in \mathds{V}} ||y-x|| d\gamma(x,y),
    \end{align*}
    where $||\cdot||$ is the Euclidean metric, and $\Gamma(p_\delta, p_\rho)$ is the set of measures over the product space $\mathds{V} \times \mathds{V}$ with marginal densities $p_\delta$ and $p_{\rho}$, respectively. Define $D_{\text{min}} \triangleq \mathit{WD}(p_\delta, p_\rho) + E[d(\rho, \delta)]$ as the expected minimum time an agent needs to travel to service a request in the map. Assume that the random variables for pickups $\rho$ and drop-offs $\delta$ are independent. Let $T\to \infty$ and the fleet size $N < E[\eta] \cdot D_{\text{min}}$. Then, the policy instantaneous assignment with reassignment policy $\pi_{\text{base}}$ is unstable according to Def.~\ref{def:routing_stability}.
    \label{theorem:coop_necessary_condition}
\end{thm}

We introduce a corollary to Theorem~\ref{theorem:coop_necessary_condition} to get it in the form that will be used in this paper.

\begin{corollary}
    Let $\mathit{WD}(p_\delta, p_\rho)$ denote the first Wasserstein distance \cite{ruschendorf1985wasserstein} between probability distributions $p_{\delta}$ and $p_{\rho}$ with support $\mathds{V}$. Define $D_{\text{min}} \triangleq \mathit{WD}(p_\delta, p_\rho) + E[d(\rho, \delta)]$. Assume that the random variables for pickups $\rho$ and drop-offs $\delta$ are independent. Let $T\to \infty$ and the fleet size $N < E[\eta] \cdot D_{\text{min}}$. Then, the random assignment policy $\hat{\pi}$ is unstable according to Def.~\ref{def:routing_stability}.
    \label{coro:coop_necessary_condition}
\end{corollary}

\begin{proof}
    This follows directly from the fact that $E[Z_{T}^{\hat{\pi}}] \geq E[Z_{T}^{\bar{\pi}}] \geq E[Z_{T}^{\pi_{\text{base}}}]$ as shown in \cite{garces2024approximate}. Since we also know from \cite{garces2024approximate} that $\lim_{T \to \infty} \frac{1}{T} E[Z_{T}^{\pi_{\text{base}}}] > E[\eta]D_{\text{min}}$, we can conclude that $\lim_{T \to \infty} \frac{1}{T} E[Z_{T}^{\hat{\pi}}] \geq \lim_{T \to \infty} \frac{1}{T} E[Z_{T}^{\pi_{\text{base}}}] > E[\eta]D_{\text{min}} > N$. And hence, the fleet size $N$ is smaller than the expected amount of time needed to service all requests (i.e. $N < \frac{1}{T} E[Z_{T}^{\hat{\pi}}]$), violating the condition for stability given in Def.~\ref{def:routing_stability}.
\end{proof}

These theoretical conditions on the fleet size, however, are not directly applicable to cases where the fleet is composed of both adversarial and cooperative agents, since adversarial agents incur a delay that increases the expected amount of time needed to service requests under any policy. For this reason, in this paper, we look for new sufficient conditions on the fleet size for which stability is still guaranteed in the presence of adversarial agents.

\section{Theoretical Results}
\label{sec:theoretical_results}

In this section, we are interested in characterizing a new sufficient condition on the fleet size $N$ that will provably guarantee that the instantaneous assignment policy $\bar{\pi}$ given in Sec.~\ref{sec:policies_def} and formally defined in algo.~\ref{alg:instantaneous_assignment_algo} is stable in the presence of adversarial agents that follow the bounded-delay model given in Def.~\ref{def:adversarial_model}. 
We first characterize the proportion of adversarial agents that will make a simple routing policy like the random assignment policy $\hat{\pi}$, defined in algo.~\ref{alg:random_assignment_algo}, unstable in the presence of bounded-delay adversarial agents. Assuming that the maximum possible proportion of adversarial agents $F_\text{max}$ is known a-priori, we then characterize a new sufficient condition on the fleet size $N$ that will provably recover stability for the the random assignment policy $\hat{\pi}$. Finally we show that since the instantaneous assignment policy $\bar{\pi}$ results in smaller times required to service requests, then the new condition on the fleet size that will guarantee stability for the random assignment policy will also guarantee stability for the instantaneous assignment policy.
Notably, this new sufficient condition shows that we can recover stability for the routing policies even if we add fewer cooperative agents than the total number of adversarial agents in the system.

\subsection{Characterizing the Proportion of Adversarial Agents That Results in Instability of Random Assignment}

In this subsection, we characterize the proportion of adversarial agents that is required to cause instability for the random assignment policy $\hat{\pi}$ defined in algo.~\ref{alg:random_assignment_algo}. To do so, we consider the following assumptions.

\begin{assumption}
    \label{assum:fleet_size}
    We assume that the system is unaware of the presence of adversarial agents in the fleet, and hence the initial fleet size $N'$ is chosen to achieve stability in the case of a fully cooperative fleet. This follows from the sufficient condition presented in Theorem~\ref{theorem:coop_sufficient_condition}.
\end{assumption}

\begin{assumption}
    \label{assum:adversarial_model}
    We assume that all adversarial agents in the system follow the bounded delay model defined in Def.\ref{def:adversarial_model}.
\end{assumption}

\begin{assumption}
    \label{assum:delay}
    We assume that the maximum adversarial delay $\Delta$ is achievable.
\end{assumption}

\begin{assumption}
    \label{assum:independence_of_pick_drop}
    We assume that the random variables for pickups $\rho$ and for drop-offs $\delta$ are independent.
\end{assumption}

Now that we have these assumptions, we move on to the formal claim of the theorem.

\begin{thm}
\label{theorem:instability}
Let $\mathit{WD}(p_\delta, p_\rho)$ denote the first Wasserstein distance \cite{ruschendorf1985wasserstein} between probability distributions $p_{\delta}$ and $p_{\rho}$ with support $\mathds{V}$. Define $D_{\text{min}} \triangleq \mathit{WD}(p_\delta, p_\rho) + E[d(\rho, \delta)]$.
Given Assumptions, \ref{assum:fleet_size} \ref{assum:adversarial_model}, \ref{assum:delay}, \ref{assum:independence_of_pick_drop}, if the proportion of adversarial agents $F$ satisfies $F > \frac{N' - E[\eta] D_{\text{min}}}{2 \Delta E[\eta]}$, then the random assignment policy $\hat{\pi}$ defined in algo.~\ref{alg:random_assignment_algo} is unstable.
\end{thm}

\begin{proof}
Define the indicator variable 
\begin{align*}
    \mathds{1}_a(\ell) = \begin{cases}
                        1 & \text{ if agent } \ell \in \mathcal{A} \\
                        0 & \text{ otherwise}
                    \end{cases}
\end{align*}
We use this indicator variable in the definition of $Z^{\text{adv}}_{T}$, the amount of time wasted by adversarial agents following the bounded-delay model defined in Def.~\ref{def:adversarial_model} during a time horizon of length $T$. The expression for $Z^{\text{adv}}_{T}$ is then given by:
\begin{align*}
    Z^{\text{adv}}_{T} = \sum_{q=1}^{R_{T}} \mathds{1}_a(\ell_{r_q}) \left( e^{\text{pick}}(\ell_{r_q}, r_q) + e^{\text{drop}}(\ell_{r_q}, r_q) \right)
\end{align*}
where $\ell_{r_q}$ is the agent assigned to request $r_q$ based on the routing policy.

To prove the statement of the theorem, we consider $Q^{\hat{\pi}}_{T}$, the effective amount of time that a fleet of $N'$ agents following policy $\hat{\pi}$ can dedicate to service requests. We write the expression for $Q^{\hat{\pi}}_{T}$ as $Q^{\hat{\pi}}_{T}= N'T - E[Z^{\text{adv}}_{T}]$, where $N'T$ is the total amount of time that a fleet of $N'$ agents could have spent servicing requests if all agents in the fleet were cooperative, while $E[Z^{\text{adv}}_{T}]$ corresponds to the expected amount of time wasted by the adversarial agents following the bounded-delay model given in Def.~\ref{def:adversarial_model}.

Combining our generalized definition of stability given in Def.~\ref{def:routing_stability} and Corollary~\ref{coro:coop_necessary_condition}, we get that if $\lim_{T \to \infty} \frac{1}{T} Q^{\hat{\pi}}_{T} < E[\eta] D_{\text{min}}$ then the random assignment policy $\hat{\pi}$ is unstable. We will use this condition to derive the threshold on the proportion of adversarial agents beyond which the random assignment policy $\hat{\pi}$ will be unstable.

We first focus on obtaining an expression for $\lim_{T \to \infty} \frac{1}{T} Q^{\hat{\pi}}_{T}$ under the presence of adversarial agents. We define $F_q$ as the proportion of adversarial agents at allocation number $q$. To do so, we consider the following:
\begin{align*}
    \lim_{T \to \infty} \frac{1}{T} Q^{\hat{\pi}}_{T} & \overset{(1)}{=} \lim_{T \to \infty} \frac{1}{T} \left( N'T - E[Z^{\text{adv}}_{T}] \right) \\
    & = N' - \lim_{T \to \infty} \frac{1}{T} E[Z^{\text{adv}}_{T}] \\
    & \overset{(2)}{=} N' - \lim_{T \to \infty} \frac{1}{T} E[ \sum_{q=1}^{R_{T}} \mathds{1}_a(\ell_{r_q}) \\
    & \qquad \qquad \left( e^{\text{pick}}(\ell_{r_q}, r_q) + e^{\text{drop}}(\ell_{r_q}, r_q) \right) ] \\
    & \overset{(3)}{\ge}  N' - \lim_{T \to \infty} \frac{1}{T} E \left[ \sum_{q=1}^{R_{T}} \mathds{1}_a(\ell_{r_q}) (2 \Delta) \right] \\
    & \overset{(4)}{=} N' - \lim_{T \to \infty} \frac{2 \Delta}{T} E \left[ \sum_{q=1}^{R_{T}}  E \left[ \mathds{1}_a(\ell_{r_q}) | R_{T} \right]\right] \\
    & \overset{(5)}{=} N' - \lim_{T \to \infty} \frac{2 \Delta}{T} E \left[ \sum_{q=1}^{R_{T}}  E \left[ E[ \mathds{1}_a(\ell_{r_q}) | F_q] \right]\right] \\
    & \overset{(6)}{=} N' - \lim_{T \to \infty} \frac{2 \Delta}{T} E \left[ \sum_{q=1}^{R_{T}}  E \left[ F_q \right]\right] \\
    & \overset{(7)}{\ge} N' - \lim_{T \to \infty} \frac{2 \Delta}{T} E \left[ \sum_{q=1}^{R_{T}}  F \right] \\
    & \overset{(8)}{=} N' - \lim_{T \to \infty} \frac{2 \Delta}{T} F E[R_{T}] \\
    & \overset{(8)}{=} N' - \lim_{T \to \infty} \frac{2 \Delta}{T} F E\left[\sum_{t=1}^{T} \eta_{t} \right] \\
    & \overset{(10)}{=} N' - \lim_{T \to \infty} \frac{2 \Delta}{T} F (T E[\eta]) \\
    & = N' - 2 \Delta F E[\eta]
\end{align*}
Where (1) follows from the definition of $Q^{\hat{\pi}}_{T}$; (2) follows from the definition of $Z^{\text{adv}}_{T}$; (3) follows from Assumption~\ref{assum:delay} and using the maximal delay adversarial agents can achieve; (4) follows from the law of total expectation and linearity of conditional expectations; (5) follows from the law of total expectation and the fact that $\mathds{1}_a(\ell_{r_q})$ does not depend on $R_T$; (6) follows from the fact that the expectation of an indicator variable is the probability of the event, and the definition of $F_q$; (7) follows form the argument stated in Appendix~\ref{appendix:fa_bound}; (8) follows from the fact that $F$ is a constant; (9) follows from the definition of $R_{T}$; (10) follows from the fact that all random variables $\eta_t$ have the same distribution and hence the same expected value of $E[\eta]$.

Now that we have an expression for $\lim_{T \to \infty} \frac{1}{T} Q^{\hat{\pi}}_{T}$, we can use it to rewrite the condition $\lim_{T \to \infty} \frac{1}{T} Q^{\hat{\pi}}_{T} < E[\eta] D_{\text{min}}$ and obtain the bounding condition on $F$ as follows:
\begin{align*}
    \lim_{T \to \infty} \frac{1}{T} Q^{\hat{\pi}}_{T} & < E[\eta] D_{\text{min}} \\
    N' - 2 \Delta F E[\eta] & < E[\eta] D_{\text{min}} \\
    N' - E[\eta] D_{\text{min}} & < 2 \Delta F E[\eta] \\
    \frac{N' - E[\eta] D_{\text{min}}}{2 \Delta E[\eta]} & < F
\end{align*}

From this, we can conclude that if the proportion of adversarial agents $F$ satisfies the condition $\frac{N' - E[\eta] D_{\text{min}}}{2 \Delta E[\eta]} < F$, then the random assignment policy $\hat{\pi}$ is unstable. This is the case since having a proportion $F$ that satisfies$\frac{N' - E[\eta] D_{\text{min}}}{2 \Delta E[\eta]} < F$ will result in satisfying the condition $\lim_{T \to \infty} \frac{1}{T} Q^{\hat{\pi}}_{T} < E[\eta] D_{\text{min}}$, which by Corollary~\ref{coro:coop_necessary_condition} implies that the policy $\hat{\pi}$ is unstable.

\end{proof}

Intuitively, the threshold on the proportion of adversarial agents $F$ is written as the surplus time that a fleet of $N'$ agents has to service requests per time step divided by the time wasted by adversarial agents per time step. The surplus time in the numerator is expressed in terms of the amount of time that the fleet with $N'$ agents can spend servicing requests in a single time step, which corresponds to $N'$ since we assume that agents can traverse one intersection per time step, minus the expected minimum amount of time needed to service requests for a single time step, which corresponds to  $E[\eta] D_{\text{min}}$. This quantity is then divided by the amount of time wasted by adversarial agents for a single time step, which corresponds to $2 \Delta E[\eta]$. For this bound, if the maximum adversarial delay increases, the proportion of adversarial agents that will result in instability becomes smaller. This happens since a larger delay increases the impact that adversarial agents will have in the system, which means that fewer adversarial agents will be enough to make the system unstable.

\subsection{New Sufficient Condition on the Fleet Size for Stability of a Routing Policy in the Presence of Adversarial Agents}

In this subsection, we derive a new sufficient condition on the fleet size $N$ in order to recover stability for both, the random assignment policy $\hat{\pi}$ and thereby the instantaneous assignment policy $\bar{\pi}$, in the presence of adversarial agents that follow the bounded-delay model. To do this, we consider the following new assumption:

\begin{assumption}
    \label{assum:knowledge_of_adversaries}
    We assume that the maximum possible proportion of adversarial agents $F_{\text{max}}$ is known in advance.
\end{assumption}

Assumption~\ref{assum:knowledge_of_adversaries} is common in the adversarial agents literature \cite{guerrero2017formations,sundaram2010distributed, pasqualetti2011consensus, saulnier2017resilient}, where the maximum tolerable number of adversaries is known a-priori.

We are interested in finding a sufficient condition on the fleet size $N$ that guarantees the stability of the instantaneous assignment policy $\bar{\pi}$ per the definition of stability given in Def.~\ref{def:routing_stability}. 
To do so, we first analyze the random assignment policy $\hat{\pi}$, which is easier to analyze since under the random assignment the request's pickup location $\rho_{r}$ and the location of the agent assigned to the request $v_{r}^{\hat{\pi}}$ become independent. 
During the analysis, we find an upper bound on $\lim_{T\to \infty} \frac{1}{T}E[Z_{T}^{\hat{\pi}}]$, where $E[Z_{T}^{\hat{\pi}}]$ is the expected time required to service requests under the random assignment policy $\hat{\pi}$. We then make sure that the fleet size selected is above this upper bound. Under these conditions, then the policy $\hat{\pi}$ becomes stable by definition. We then move on to show that $E[Z_{T}^{\bar{\pi}}]$, the expected time required to service requests under the instantaneous assignment policy $\bar{\pi}$, is smaller than $E[Z_{T}^{\hat{\pi}}]$. We then show that $\lim_{T\to \infty} \frac{1}{T}E[Z_{T}^{\bar{\pi}}] \leq \lim_{T\to \infty} \frac{1}{T} E[Z_{T}^{\hat{\pi}}]$, and hence the upper bound found on $\lim_{T\to \infty} \frac{1}{T}E[Z_{T}^{\hat{\pi}}]$ also applies to $\lim_{T\to \infty} \frac{1}{T}E[Z_{T}^{\bar{\pi}}]$. From this, we can conclude that $\bar{\pi}$ is also stable by definition as long as the fleet size $N$ is chosen above the derived upper bound. We present the formal claim for the sufficient condition on $N$ for the random assignment policy $\hat{\pi}$ in the following theorem.

\begin{thm}
\label{theorem:new_sufficient_condition_RA}
Given Assumptions~\ref{assum:adversarial_model}, \ref{assum:delay}, and \ref{assum:knowledge_of_adversaries}, if the fleet size $N$ satisfies $N \geq E\left[ \eta  \right]{D_{\max }} + 2\Delta E\left[ \eta  \right] {F_{\text{max}}}$ as $T \to \infty$,
then the random assignment policy $\hat{\pi}$ is stable according to Def.~\ref{def:routing_stability}, where $D_{\text{max}} \triangleq \max\{E[d(\xi, \rho)], E[d(v_{\text{rand}}, \rho)]\} + E[d(\rho, \delta)]$ corresponds to the maximum expected time needed to service a request, and the random variable $v_{\text{rand}}$ with support $\mathds{V}$ represents the location of a random agent that gets assigned to a request after that agent has previously serviced a different request.
\end{thm}

\begin{proof}
Consider ${Z_{T}^{\hat{\pi}}}$, the total time needed for the random assignment policy $\hat{\pi}$ to service all requests that enter the system during a time horizon of length $T$. We can rewrite ${Z_{T}^{\hat{\pi}}}$ as,
\begin{align}
    Z_{T}^{\hat{\pi}} = \sum_{q=1}^{R_T} W_{r_q}^{\hat{\pi}} \nonumber
\end{align}
so that we consider requests based on their index irrespective of the time at which the request enter the system.
We reindex requests such that requests that are assigned to agents that haven't serviced any requests come before requests that are assigned to agents that have already serviced one or more requests. Using this, we expand $Z_{T}^{\hat{\pi}}$ as follows:

\begin{align*}
    {Z_{T}^{\hat{\pi}}} & = \sum\limits_{b_1 = 1}^{{n_{c}^{\text{init}}}} {{W^{\hat \pi}_{{r_{b_1}}}}}  + \sum\limits_{b_2 = {n_{c}^{\text{init}}} + 1}^{{n_{c}^{\text{init}}} + {n_{c}^{\text{next}}}} {{W^{\hat \pi}_{{r_{b_2}}}}} \nonumber \\
& + \sum\limits_{b_3 = {n_{c}^{\text{init}}} + {n_{c}^{\text{next}}} + 1}^{{n_{c}^{\text{init}}} + {n_{c}^{\text{next}}} + {n_{a}^{\text{init}}}} {{W^{\hat \pi}_{{r_{b_3}} }}}  + \sum\limits_{b_4 = {n_{c}^{\text{init}}} + {n_{c}^{\text{next}}} + {n_{a}^{\text{init}}} + 1}^{{R_{_T}}} {{W^{\hat \pi}_{{r_{b_4}} }}}
\end{align*}

Where $n_{c}^{\text{init}}$, $n_{c}^{\text{next}}$ denote the number of requests that are allocated to cooperative agents that start servicing a request from their initial locations or from a location after having serviced a previous request, respectively. Similarly, $n_{a}^{\text{init}}$, $n_{a}^{\text{next}}$ denote the number of requests that are allocated to adversarial agents that start servicing a request from their initial locations or from a location after having serviced a previous request, respectively. It is worth noting that $n_{a}^{\text{next}}= {{R_T} - {n_{c}^{\text{init}}} - {n_{c}^{\text{next}}} - {n_{a}^{\text{init}}}}$.
Taking the expectation over ${Z_{T}^{\hat{\pi}}}$ and using the linearity of expectations yields,

\begin{equation}
\begin{array}{l}
E\left[ {{{Z_{T}^{\hat{\pi}}}}} \right]  = E\left[ {\sum\limits_{b_1 = 1}^{{n_{c}^{\text{init}}}} {{W^{\hat \pi}_{{r_{b_1}} }}} } \right] + E\left[ {\sum\limits_{b_2 = {n_{c}^{\text{init}}} + 1}^{{n_{c}^{\text{init}}} + {n_{c}^{\text{next}}}} {{W^{\hat \pi}_{{r_{b_2}} }}} } \right] +\\ 
E\left[ {\sum\limits_{b_3 = {n_{c}^{\text{init}}} + {n_{c}^{\text{next}}} + 1}^{{n_{c}^{\text{init}}} + {n_{c}^{\text{next}}} + {n_{a}^{\text{init}}}} {{W^{\hat \pi}_{{r_{b_3}} }}} } \right] + E\left[ {\sum\limits_{b_4 = {n_{c}^{\text{init}}} + {n_{c}^{\text{next}}} + {n_{a}^{\text{init}}} + 1}^{{R_{_T}}} {{W^{\hat \pi}_{{r_{b_4}} }}} } \right] \\ \mathop  = \limits^{(1)}  
E\left[ {E\left[ {\left. {\sum\limits_{b_1 = 1}^{{n_{c}^{\text{init}}}} {{W^{\hat \pi}_{{r_{b_1}} }}} } \right|{n_{c}^{\text{init}}}} \right]} \right] + \\ E\left[ {E\left[ {\left. {\sum\limits_{b_2 = {n_{c}^{\text{init}}} + 1}^{{m_{c}^{\text{init}}} + {n_{c}^{\text{next}}}} {{W^{\hat \pi}_{{r_{b_2}} }}} } \right|{n_{c}^{\text{init}}},{n_{c}^{\text{next}}}} \right]} \right] + \\ 
E\left[ {E\left[ {\left. {\sum\limits_{b_3 = {n_{c}^{\text{init}}} + {n_{c}^{\text{next}}} + 1}^{{n_{c}^{\text{init}}} + {n_{c}^{\text{next}}} + {n_{a}^{\text{init}}}} {{W^{\text{adv}}_{r_{b_3}}}}} \right|{n_{c}^{\text{init}}},{n_{c}^{\text{next}}},{n_{a}^{\text{init}}}} \right]} \right] + \\ E\left[ {E\left[ {\left. {\sum\limits_{b_4 = {n_{c}^{\text{init}}} + {n_{c}^{\text{next}}} + {n_{a}^{\text{init}}} + 1}^{{R_T}} {W^{\text{adv}}_{r_{b_4}}}} \right|{n_{c}^{\text{init}}},{n_{c}^{\text{next}}},{n_{a}^{\text{init}}},{R_T}} \right]} \right] \\  \mathop = \limits^{(2)}
 E\left[ {\sum\limits_{b_1 = 1}^{{n_{c}^{\text{init}}}} {E\left[ {\left. {{W^{\hat \pi}_{{r_{b_1}} }}} \right|{n_{c}^{\text{init}}}} \right]} } \right] + \\ E\left[ {\sum\limits_{b_2 = {n_{c}^{\text{init}}} + 1}^{{n_{c}^{\text{init}}} + {n_{c}^{\text{next}}}} {E\left[ {\left. {{W^{\hat \pi}_{{r_{b_2}} }}} \right|{n_{c}^{\text{init}}},{n_{c}^{\text{next}}}} \right]} } \right]  + \\
E\left[ {\sum\limits_{b_3 = {n_{c}^{\text{init}}} + {n_{cm}} + 1}^{{n_{c}^{\text{init}}} + {n_{c}^{\text{next}}} + {n_{a}^{\text{init}}}} {E\left[ {\left. {{W^{\text{adv}}_{r_{b_3}}}} \right|{n_{c}^{\text{init}}},{n_{c}^{\text{next}}},{n_{a}^{\text{init}}}} \right]} } \right] + \\ E\left[ {\sum\limits_{b_4 = {n_{c}^{\text{init}}} + {n_{c}^{\text{next}}} + {n_{a}^{\text{init}}} + 1}^{{R_T}} {E\left[ {\left. {{W^{\text{adv}}_{r_{b_4}}}} \right|{n_{c}^{\text{init}}},{n_{c}^{\text{next}}},{n_{a}^{\text{init}}},{R_T}} \right]} } \right]
\end{array} 
\label{e101}  
\end{equation}

Where in $(1)$ we use the law of total expectations and in $(2)$ the linearity of conditional expectations to move the expectation operator inside the summation. 
We define $W^{\hat{\pi}}_{r_{\text{init}}} = d(v_{r_{\text{init}}}^{\hat{\pi}}, \rho_{r_{\text{init}}}) + d(\rho_{r_{\text{init}}}, \delta_{r_{\text{init}}})$, where the random variables $v_{r_{\text{init}}}^{\hat{\pi}}$, $\rho_{r_{\text{init}}}$, and $\delta_{r_{\text{init}}}$ have the same distributions as $\xi$, $\rho$, and $\delta$, respectively. We also define $W_{r_{\text{next}}}^{\hat{\pi}} = d(v_{r_{\text{next}}}^{\hat{\pi}}, \rho_{r_{\text{next}}}) + d(\rho_{r_{\text{next}}}, \delta_{r_{\text{next}}})$, where the random variables $v_{r_{\text{next}}}^{\hat{\pi}}$, $\rho_{r_{\text{next}}}$, and $\delta_{r_{\text{next}}}$ have the same distributions as $v_{\text{rand}}$, $\rho$, and $\delta$, respectively.
Similarly, we define ${W^{\text{adv}}_{r_{\text{init}}}} = d_{\text{adv}}(v_{r_{\text{init}}}^{\hat{\pi}}, \rho_{r_{\text{init}}}) + d_{\text{adv}}(\rho_{r_{\text{init}}}, \delta_{r_{\text{init}}})$ where the random variables $v_{r_{\text{init}}}^{\hat{\pi}}$, $\rho_{r_{\text{init}}}$, and $\delta_{r_{\text{init}}}$ have the same distributions as $\xi$, $\rho$, and $\delta$, respectively. Finally, we define  ${W^{\text{adv}}_{r_{\text{next}}}} = d_{\text{adv}}(v_{r_{\text{next}}}^{\hat{\pi}}, \rho_{r_{\text{next}}}) + d_{\text{adv}}(\rho_{r_{\text{next}}}, \delta_{r_{\text{next}}})$, where the random variables $v_{r_{\text{next}}}^{\hat{\pi}}$, $\rho_{r_{\text{next}}}$, and $\delta_{r_{\text{next}}}$ have the same distributions as $v_{\text{rand}}$, $\rho$, and $\delta$, respectively.
Under policy $\hat{\pi}$, we have the following: ${{W^{\hat \pi}_{{r_{b_1}} }}}$ is independent of $n_{c}^{\text{init}}$; ${{W^{\hat \pi}_{{r_{b_2}} }}}$ is independent of $n_{c}^{\text{init}}, n_{c}^{\text{next}}$; ${W^{\text{adv}}_{r_{b_3}}}$ is independent of $n_{c}^{\text{init}},n_{c}^{\text{next}}, n_{a}^{\text{init}}$; and ${W^{\text{adv}}_{r_{b_4}}}$ is independent of $n_{c}^{\text{init}}, n_{c}^{\text{next}}, n_{a}^{\text{init}},{R_T}$. Using these independent relations, equation \ref{e101} yields, 

\begin{equation}
\begin{array}{l}
E\left[ {\sum\limits_{b_1 = 1}^{{n_{c}^{\text{init}}}} {E\left[ {{W^{\hat \pi}_{{r_{b_1}} }}} \right]} } \right] + E\left[ {\sum\limits_{b_2 = {n_{c}^{\text{init}}} + 1}^{{n_{c}^{\text{init}}} + {n_{c}^{\text{next}}}} {E\left[ {{W^{\hat \pi}_{{r_{b_2}} }}} \right]} } \right] + \\
E\left[ {\sum\limits_{b_3 = {n_{c}^{\text{init}}} + {n_{c}^{\text{next}}} + 1}^{{n_{c}^{\text{init}}} + {n_{c}^{\text{next}}} + {n_{a}^{\text{init}}}} {E\left[ {{W^{\hat \pi}_{{r_{b_3}} }}} \right]} } \right] + \\ E\left[ {\sum\limits_{b_4 = {n_{c}^{\text{init}}} + {n_{c}^{\text{next}}} + {n_{a}^{\text{init}}} + 1}^{{R_T}} {E\left[ {{W^{\hat \pi}_{{r_{b_4}} }}} \right]} } \right] \\ \mathop  = \limits^{(3)}
E\left[ {{n_{c}^{\text{init}}}} \right]E\left[ {{W^{\hat \pi}_{{r_{\text{init}}} }}} \right]  + E\left[ {{n_{c}^{\text{next}}}} \right]E\left[ {{W^{\hat \pi}_{{r_{\text{next}}} }}} \right] +\\ E\left[ {{n_{a}^{\text{init}}}} \right]E\left[ {W^{\text{adv}}_{r_{\text{init}}}} \right] + \\ E\left[ {{R_T} - {n_{c}^{\text{init}}} - {n_{c}^{\text{next}}} - {n_{a}^{\text{init}}}} \right]E\left[ {W^{\text{adv}}_{r_{\text{next}}}} \right]
\end{array} 
\label{e220} 
\end{equation}
Where in $(3)$ we use the fact that all $W^{\hat \pi}_{r_{b_1}}$ have the same distribution as $W^{\hat \pi}_{r_{\text{init}}}$, all $W^{\hat \pi}_{r_{b_2}}$ have the same distribution as $W^{\hat \pi}_{r_{\text{next}}}$, all $W^{\hat \pi}_{r_{b_3}}$ have the same distribution as ${W^{\text{adv}}_{r_{\text{init}}}}$, and all of $W^{\hat \pi}_{r_{b_4}}$ have the same distribution as ${W^{\text{adv}}_{r_{\text{next}}}}$. We will analyze each of the terms independently and then combine the results. 

We obtain an upper bound for  
$E\left[ {{W^{\hat \pi}_{{r_{\text{init}}} }}} \right]$ using the following analysis,
\begin{equation}
\begin{array}{l}
E\left[ {{W^{\hat \pi}_{{r_{\text{init}}} }}} \right] = E\left[ {d\left( {v_{r_{\text{init}}}^{\hat{\pi}},{\rho _{{r_{\text{init}}}}}} \right)} \right] + E\left[ {d\left( {{\rho _{{r_{\text{init}}}}},{\delta _{{r_{\text{init}}}}}} \right)} \right] \\=  E\left[ {d\left( {\xi ,\rho } \right)} \right] + E\left[ {d\left( {\rho ,\delta } \right)} \right]  \\\le
 \max \left\{ {E\left[ {d\left( {\xi ,\rho } \right)} \right],E\left[ {d\left( {{v_{\text{rand}}},\rho } \right)} \right]} \right\} + E\left[ {d\left( {\rho ,\delta } \right)} \right] = {D_{\max }}
\end{array}   
\label{e221} 
\end{equation}

Using similar arguments we can get an upper bound for $E\left[ {{W^{\hat \pi}_{{r_{\text{next}}} }}} \right]$ as follows:
\begin{equation}
\begin{array}{l}
E\left[ {{W^{\hat \pi}_{{r_{\text{next}}} }}} \right]= E\left[ {d\left( {v_{r_{\text{next}}}^{\hat{\pi}},{\rho _{{r_{\text{next}}}}}} \right)} \right] + E\left[ {d\left( {{\rho _{{r_{\text{next}}}}},{\delta _{{r_{\text{next}}}}}} \right)} \right]  \\ 
 = E\left[ {d\left( {{v_{\text{rand}}},\rho } \right)} \right] + E\left[ {d\left( {\rho ,\delta } \right)} \right]  \\\le
\max \left\{ {E\left[ {d\left( {\xi ,\rho } \right)} \right],E\left[ {d\left( {{v_{\text{rand}}},\rho } \right)} \right]} \right\} + E\left[ {d\left( {\rho ,\delta } \right)} \right] = {D_{\max }}
\end{array}  
\label{e222} 
\end{equation}
Now if we consider the values for the adversarial agents, we can get an upper bound to  $E\left[ {W^{\text{adv}}_{r_{\text{init}}}} \right]$ as follows:
\begin{equation}
\begin{array}{l}
E\left[ {W^{\text{adv}}_{r_{\text{init}}}} \right] = E\left[ {d_{\text{adv}}\left( {v_{r_{\text{init}}}^{\hat{\pi}},{\rho _{{r_{\text{init}}}}}} \right)} \right] + E\left[ {d_{\text{adv}}\left( {{\rho _{{r_{\text{init}}}}},{\delta _{{r_{\text{init}}}}}} \right)} \right] \\=  E\left[ {d_{\text{adv}}\left( {\xi ,\rho } \right)} \right] + E\left[ {d_{\text{adv}}\left( {\rho ,\delta } \right)} \right]\\ \leq E\left[ {d\left( {\xi ,\rho } \right)} \right] + E\left[ {d\left( {\rho ,\delta } \right)} \right] + 2\Delta  \\ \le
 \max \left\{ {E\left[ {{d}\left( {\xi ,\rho } \right)} \right],E\left[ {{d}\left( {{v_{\text{rand}}},\rho } \right)} \right]} \right\} + E\left[ {{d_{\text{adv}}}\left( {\rho ,\delta } \right)} \right] + 2\Delta\\=  {D_{\max }} + 2\Delta
\end{array}
\label{e223}    
\end{equation}
Similarly, we can get an upper bound for $E\left[ {W^{\text{adv}}_{r_{\text{next}}}} \right]$ as follows:
\begin{equation}
\begin{array}{l}
E\left[ {W^{\text{adv}}_{r_{\text{next}}}} \right]= E\left[ {d_{\text{adv}}\left( {v_{r_{\text{next}}}^{\hat{\pi}},{\rho _{{r_{\text{next}}}}}} \right)} \right] +  E\left[ {d_{\text{adv}}\left( {{\rho _{{r_{\text{next}}}}},{\delta _{{r_{\text{next}}}}}} \right)} \right] \\= E\left[ {d_{\text{adv}}\left( {{v_{\text{rand}}},\rho } \right)} \right] + E\left[ {d_{\text{adv}}\left( {\rho ,\delta } \right)} \right] \\ \leq E\left[ {d\left( {{v_{\text{rand}}},\rho } \right)} \right] + E\left[ {d\left( {\rho ,\delta } \right)} \right] +2\Delta \\ \le \max \left\{ {E\left[ {{d}\left( {\xi ,\rho } \right)} \right],E\left[ {{d}\left( {{v_{\text{rand}}},\rho } \right)} \right]} \right\} + E\left[ {{d_{\text{adv}}}\left( {\rho ,\delta } \right)} \right]+2\Delta \\ = {D_{\max }} + 2\Delta
\label{e106}   
\end{array}
\end{equation} 
Since the expected number of requests that are assigned to cooperative agents starting from their initial locations is upper bounded by the number of cooperative agents, we have that $E\left[ {{n_{c}^{\text{init}}}} \right] \leq |\mathcal{C}|$ and using the same reasoning the expected number of requests that are assigned to adversarial agents starting from their initial locations is upper bounded by the number of adversarial agents such that $E\left[ {{n_{a}^{\text{init}}}} \right] \leq |\mathcal{A}|$. Hence, using the developed inequalities of equations \ref{e221},\ref{e222},\ref{e223}, and \ref{e106}, we have,
\begin{equation}
\begin{array}{l}
E\left[  {Z_{T}^{\hat{\pi}}} \right] \le |\mathcal{C}|{D_{\max }} + E\left[ {{n_{c}^{\text{next}}}} \right]{D_{\max }} + |\mathcal{A}|\left( {{D_{\max }} + 2\Delta } \right) + \\ E\left[ {{R_T} - {n_{c}^{\text{init}}} - {n_{c}^{\text{next}}} - {n_{a}^{\text{init}}}} \right]\left( {{D_{\max }} + 2\Delta } \right)
\end{array}
\label{e107}   
\end{equation}
We now divide by $T$ and take the limit as $T \to \infty$ of the following terms that appear on the right hand side of equation \ref{e107},
\begin{equation}
\mathop {\lim }\limits_{T \to \infty } \frac{{{|\mathcal{C}|}{D_{\max }}}}{T} \to 0
\label{e108}     
\end{equation}
Where equation \ref{e108} holds since both $|\mathcal{C}|$ and $D_{\text{max}}$ are finite,
\begin{equation}
\mathop {\lim }\limits_{T \to \infty } \frac{{|\mathcal{A}|}\left( {{D_{\max }} + 2\Delta } \right)}{T} \to 0
\label{e109}     
\end{equation}
Where equation \ref{e109} holds since $|\mathcal{A}|, D_{\text{max}}, \Delta$ are finite,
\begin{equation}
\mathop {\lim }\limits_{T \to \infty } \frac{{E\left[ {{n_{c}^{\text{init}}}} \right]{\left( {{D_{\max }} + 2\Delta } \right)}}}{T} \leq \frac{|\mathcal{C}|{\left( {{D_{\max }} + 2\Delta } \right)}}{T} \to 0
\label{e110}     
\end{equation}
Where equation \ref{e110} holds since $|\mathcal{A}|, D_{\text{max}}, \Delta$ are finite,
\begin{equation}
\mathop {\lim }\limits_{T \to \infty } \frac{{E\left[ {{n_{a}^{\text{init}}}} \right]{\left( {{D_{\max }} + 2\Delta } \right)}}}{T} 
\leq \frac{|\mathcal{A}|{\left( {{D_{\max }} + 2\Delta } \right)}}{T} \to 0
\label{e111}     
\end{equation}
Using the results obtained from equations \ref{e108}, \ref{e109}, and \ref{e110}, we get:
\begin{equation}
\begin{array}{l}
\mathop {\lim }\limits_{T \to \infty } \frac{E\left[  {Z_{T}^{\hat{\pi}}} \right]}{T}  \\ \leq \mathop {\lim }\limits_{T \to \infty } \frac{1}{T}\left( {E\left[ {{n_{c}^{\text{next}}}} \right]{D_{\max }} + E\left[ {{R_T} - {n_{c}^{\text{next}}}} \right]{\left( {{D_{\max }} + 2\Delta } \right)}} \right) 
\end{array}
\label{e112}    
\end{equation}
Denote by time $t'$ the time at which, almost surely, all cooperative agents were assigned at least one request by the random assignment policy and hence are no longer at their initial location. We know that such a time exists since the number of cooperative agents is finite and the time $ T \to \infty$, which means that with probability $1.0$ all the cooperative agents will be assigned to at least one request and will move from their initial positions. Furthermore, based on the known number of cooperative agents, $|\mathcal{C}|$, time $t'$ can be analytically calculated as a solution to the coupon collector's problem \cite{mitzenmacher2017probability} using the formula $E\left[t'\right] = \frac{{{|\mathcal{C}|}}}{{{|\mathcal{C}|}}} + \frac{{{|\mathcal{C}|}}}{{{|\mathcal{C}|} - 1}} + ... + \frac{{{|\mathcal{C}|}}}{1}$. Denote by ${\eta _{c,{t_f}}}$ the number of requests that are assigned to cooperative agents at time $t_f$. Using this definition, we can write $n_{c}^{\text{next}} = \sum\limits_{{t_f} = t'}^T {{\eta _{c,{t_f}}}}$. We now analyze the following expression:
\begin{equation}
\begin{array}{l}
\mathop {\lim }\limits_{T \to \infty } \frac{{E\left[ {{n_{c}^{\text{next}}}} \right]}}{T} \mathop  = \limits^{(4)} \mathop {\lim }\limits_{T \to \infty } \frac{1}{T} {E\left[ {\sum\limits_{{t_f} = t'}^T {{\eta _{c,{t_f}}}}  } \right]}   \\ \mathop  = \limits^{(5)}  \mathop {\lim }\limits_{T \to \infty } \frac{1}{T}{E\left[ {E\left[ {\left. {\sum\limits_{{t_f} = t'}^T {{\eta _{c,{t_f}}}} } \right|t'} \right]} \right]}   \\  \mathop  = \limits^{(6)} \mathop {\lim }\limits_{T \to \infty } \frac{1}{T}{E\left[ {\sum\limits_{{t_f} = t'}^T {E\left[ {\left. {{\eta _{c,{t_f}}}} \right|t'} \right]} } \right]} \\ \mathop  =  \limits^{(7)} \mathop {\lim }\limits_{T \to \infty } \frac{1}{T}{E\left[ {\sum\limits_{{t_f} = t'}^T {E\left[ {{\eta _{c,{t_f}}}} \right]} } \right]} 
\end{array}
\label{e113}    
\end{equation}
In $(4)$ we express $n_{c}^{\text{next}}$ according to the sum of the number of requests that are assigned to cooperative agents at each time step following time $t'$. In $(5)$ we use the law of total expectations. In $(6)$ we use the linearity of conditional expectations and in $(7)$ we use the fact that ${\eta _{c,{t_f}}}$ is independent of $t'$. 

We further develop the expectation inside the sum,  ${E\left[ {{\eta _{c,{t_f}}}} \right]}$. We define the indicator variable
\begin{align*}
    \mathds{1}_c(\ell) = \begin{cases}
                        1 & \text{ if agent } \ell \in \mathcal{C} \\
                        0 & \text{ otherwise}
                    \end{cases}
\end{align*}
We also define $\ell_{r^{\hat{\pi}}_q}$ as the agent that gets assigned to request $r_q$ according to policy $\hat{\pi}$, $F_{c, k}$ as the proportion of cooperative agents at allocation $k$, and by $F_c$ the overall proportion of cooperative agents in the fleet, where $F_c= 1- F$. We look at the expectation of assigned requests at time $t=t_f$,
\begin{align}
    E\left[ {{\eta _{c,{t_f}}}} \right] & \overset{(8)}{=} E\left[ {\sum\limits_{k = 1}^{{\eta_{t_f}}} \mathds{1}_c(\ell_{r^{\hat{\pi}}_q}) } \right] \nonumber \nonumber \\
    & \overset{(9)}{=} E\left[ {E\left[ {\left. {\sum\limits_{k = 1}^{{\eta _{{t_f}}}} \mathds{1}_c(\ell_{r^{\hat{\pi}}_q}) } \right|{\eta _{{t_f}}}} \right]} \right] \nonumber \\
    & \overset{(10)}{=} E\left[ {\sum\limits_{k = 1}^{{\eta_{t_f}}} {E\left[ {\left. \mathds{1}_c(\ell_{r^{\hat{\pi}}_q}) \right|{\eta_{t_f}}} \right]} } \right] \nonumber \\
    & \overset{(11)}{=} E\left[ {\sum\limits_{k = 1}^{{\eta_{t_f}}} {E\left[ E[\mathds{1}_c(\ell_{r^{\hat{\pi}}_q}) | F_{c,k}] \right]} } \right] \nonumber \\
    & \overset{(12)}{=} E\left[ {\sum\limits_{k = 1}^{{\eta_{t_f}}} {E\left[ F_{c,k} \right]} } \right] \nonumber \\
    & \overset{(13)}{\geq} E\left[ {\sum\limits_{k = 1}^{{\eta_{t_f}}} {F_c} } \right] \nonumber  \\
    & \overset{(14)}{=} {F_c} E\left[ {{\eta_{t_f}}} \right] \nonumber \\
    & \overset{(15)}{= } {F_c}E\left[ \eta  \right]
    \label{e225}
\end{align}

 In $(8)$ we describe $\eta _{c,{t}}$ in terms of $\mathds{1}_c(\ell_{r^{\hat{\pi}}_q, })$. In $(9)$ we use the law of total expectations, then in $(10)$ we use the linearity of conditional expectations. In $(11)$ we use the law of total expectations, then in $(12)$ we use the fact that the expectation of the indicator function is the probability that an agent is cooperative in the fleet at assignment $k$, which corresponds to $F_{c,k}$. In $(13)$ we use the argument stated in Appendix~\ref{appendix:fa_bound}. 
 In $(14)$ we use the fact that $F_c$ is a constant and hence can be moved outside of the expectation, and at $(15)$ we assume that the requests that enter the system at each time step are all identically distributed with the same distribution of the random variable $\eta$. Therefore, equation \ref{e113} yields,
\begin{equation}
\begin{array}{l}
\mathop {\lim }\limits_{T \to \infty } \frac{1}{T}\left( {\sum\limits_{{t_f} = t'}^T {E\left[ {{\eta _{c,{t_f}}}} \right]} } \right) \mathop \geq \limits^{(15)} \mathop {\lim }\limits_{T \to \infty } \frac{1}{T}E\left[ {\sum\limits_{{t_f} = t'}^T {{F_c}E\left[ \eta  \right]} } \right] \\ \mathop =\limits^{(16)} \mathop {\lim }\limits_{T \to \infty } \frac{1}{T}\left( {{F_c}E\left[ \eta  \right]T - {F_c}E\left[ \eta  \right]E\left[ {t'} \right]} \right) \mathop =\limits^{(17)} {F_c}E\left[ \eta  \right]
\end{array}
\label{e230}
\end{equation}
Where in $(15)$ we substitute the result from equation \ref{e225}, in $(16)$ we use the fact that $E[T]=T$ since $T$ is a constant and in $(17)$ the second term goes to zero as $T \to \infty$ since ${{F_c}, E\left[ \eta  \right], E\left[ {t'} \right]}$ are all finite.
We now analyze the second remaining expression in equation \ref{e112} as follows,
\begin{equation}
\begin{array}{l}
E\left[ {{R_T} - {n_{c}^{\text{next}}}} \right]{\left( {{D_{\max }} + 2\Delta } \right)}\\ \mathop =\limits^{(18)} 
\left( {E\left[ {{R_T}} \right] - E\left[ {{n_{c}^{\text{next}}}} \right]} \right)\left( {{D_{\max }} + 2\Delta } \right) \\ \mathop =\limits^{(19)} \left( {{D_{\max }} + 2\Delta } \right)\left( {E\left[ {\sum\limits_{{t_i} = 1}^T {{\eta _{{t_i}}}} } \right] - E\left[ {{n_{c}^{\text{next}}}} \right]} \right)\\
 \mathop =\limits^{(20)} \left( {{D_{\max }} + 2\Delta } \right)\left( {\sum\limits_{{t_i} = 1}^T {E\left[ \eta  \right] - E\left[ {{n_{c}^{\text{next}}}} \right]} } \right) \\\mathop \leq\limits^{(21)}  \left( {{D_{\max }} + 2\Delta } \right)\left( {TE\left[ \eta  \right] - {F_c}E\left[ \eta  \right]T + {F_c}E\left[ \eta  \right]E\left[ {t'} \right]} \right)
\end{array}
\label{e115}   
\end{equation}
Where in $(18)$ we use the linearity of expectations, in $(19)$ we use the definition of $R_T$, in $(20)$ we use the linearity of the expectation and the sum operator and the fact that all requests that enter the system at each time step are all identically distributed with the same distribution of the random variable $\eta$. In $(21)$ we use the inequality  $E[n_{c}^{\text{next}}] \geq {{F_c}E\left[ \eta  \right]T - {F_c}E\left[ \eta  \right]E\left[ {t'} \right]}$ obtained in equation \ref{e230}.
Taking the limit as $T \to \infty$ for the expression in (\ref{e115}) yields,
\begin{equation}
\begin{array}{l}
\mathop {\lim }\limits_{T \to \infty } \frac{1}{T}\left( {{D_{\max }} + 2\Delta } \right)\left( {TE\left[ \eta  \right] - {F_c}E\left[ \eta  \right]T + {F_c}E\left[ \eta  \right]E\left[ {t'} \right]} \right) \\ \mathop =\limits^{(22)} E\left[ \eta  \right]\left( {{D_{\max }} + 2\Delta } \right)\left( {1 - {F_c}} \right)
\end{array}
\label{e121}   
\end{equation}
Where $(22)$ follows from the fact that as $T \to \infty$ the term  ${F_c}E\left[ \eta  \right]E\left[ {t'} \right]$ goes to zero since  ${{F_c}, E\left[ \eta  \right], E\left[ {t'} \right]}$ are all finite. Therefore we obtain that,
\begin{equation}
\begin{array}{l}
\mathop {\lim }\limits_{T \to \infty } \frac{E\left[  {Z_{T}^{\hat{\pi}}} \right]}{T} \mathop \le\limits^{(23)} E\left[ \eta  \right]{F_c}{D_{\max }} + E\left[ \eta  \right]\left( {{D_{\max }} + 2\Delta } \right)\left( {1 - {F_c}} \right) \\ \mathop =\limits^{(24)} E\left[ \eta  \right]{F_c}{D_{\max }} + E\left[ \eta  \right]{D_{\max }} + 2\Delta E\left[ \eta  \right] - E\left[ \eta  \right]{F_c}{D_{\max }} \\ - E\left[ \eta  \right]2\Delta {F_c} \mathop =\limits^{(25)} 
E\left[ \eta  \right]{D_{\max }} + 2\Delta E\left[ \eta  \right] - E\left[ \eta  \right]2\Delta {F_c} \\ \mathop =\limits^{(26)} E\left[ \eta  \right]{D_{\max }} + 2\Delta E\left[ \eta  \right]\left( {1 - {F_c}} \right) \\
=
E\left[ \eta  \right]{D_{\max }} + 2\Delta E\left[ \eta  \right]F \mathop \le\limits^{(27)} E\left[ \eta  \right]{D_{\max }} + 2\Delta E\left[ \eta  \right]{F_{\text{max}}}
\end{array}
\label{e116}    
\end{equation}
Where $(23)$ follows from substituting the results obtained in equations \ref{e230} and \ref{e121} into equation \ref{e112}, $(24)-(26)$ follow from rearranging terms and $(27)$ follows from using Assumption \ref{assum:knowledge_of_adversaries}.  

Now that we have found an upper bound on $\mathop {\lim }\limits_{T \to \infty } \frac{E\left[  {Z_{T}^{\hat{\pi}}} \right]}{T}$, if we choose $N$ such that
\begin{equation}
N \geq E\left[ \eta  \right]{D_{\max }} + 2\Delta E\left[ \eta  \right]{F_{\text{max}}}
\label{e117}    
\end{equation}
We will have that $\mathop {\lim }\limits_{T \to \infty } \frac{{E\left[  {Z_{T}^{\hat{\pi}}} \right]
}}{T} \leq N$, and policy $\hat{\pi}$ will be stable by definition Def.~\ref{def:routing_stability}.
\end{proof}

We can express the probability distribution $p_{v_{\text{rand}}}$ for ${v_{\text{rand}}}$ using the marginalization of the probabilities of the previous pickup-dropoff combinations, and therefore all the terms in $D_{\text{max}}$ can be calculated using historical request data. Intuitively, the sufficient condition is stating that the amount of time that a fleet of $N$ agents can spend servicing requests per time step (which corresponds to $N$ because of our assumption that agents travel one edge per time step) has to be greater than the maximum expected time needed to service requests per time step $E\left[ \eta  \right]{D_{\max }}$ plus the maximum amount of time wasted by adversarial agents per time step $2\Delta E\left[ \eta  \right] {F_{\text{max}}}$.

We use the result from Theorem~\ref{theorem:new_sufficient_condition_RA} to show that the same conditions on the fleet size $N$ will guarantee stability for the instantaneous assignment policy $\bar{\pi}$. We formalize this claim in the following theorem. 

\begin{thm}
\label{theorem:new_sufficient_condition_IA}
Given Assumptions~\ref{assum:adversarial_model}, \ref{assum:delay}, and \ref{assum:knowledge_of_adversaries}, if the fleet size $N$ satisfies the condition given in Theorem~\ref{theorem:new_sufficient_condition_RA}, then the instantaneous assignment policy $\bar{\pi}$ defined in algo.~\ref{alg:instantaneous_assignment_algo} is stable according to Def.~\ref{def:routing_stability}.
\end{thm}

\begin{proof}
The above analysis considered random assignments of agents to requests using policy $\hat \pi$. We next show that using instantaneous assignment as the base policy, denoted by $\bar \pi$, instead of the random assignment policy results in lesser expected times to service all requests, and hence constitutes a better base policy.  Denote by ${Z_{T}^{\bar{\pi}}}$ the time needed for an agent fleet to service all requests that enter the system during a horizon of length $T$, where requests are assigned to agents based on policy $\bar \pi$. 

We define as $d'$, the time it takes for any agent to travel between locations. This distance is the shortest path for cooperative agents and the delayed distance given in Def.~\ref{def:adversarial_model} for the adversarial agents.

The total time that the fleet travels to service all requests from the start of the mission until time $T$, when requests are assigned according to the instantaneous assignment base policy, $\bar \pi$, is hence given by,
\begin{equation}
\begin{array}{l}
{Z_{T}^{\bar{\pi}}} = {\sum\limits_{t = 1}^T {\mathop {\min }\limits_{{ v_{R_{t-1}+1}},...,{v_{R_t}}} \sum\limits_{q = {{R_{t - 1}} + 1}}^{{R_{t}}} {d'\left( {v_{r_{q}},{\rho _{{r_q}}}} \right) + d'\left( {{\rho _{{r_q}}},{\delta _{{r_q}}}} \right)} } }
\end{array}
\label{e120} 
\end{equation}
where $v_{r_{q}}$ corresponds to the location of the agent that got assigned to request $r_q$ based on the minimization.
The total time the fleet travels to service all requests from the start of the mission until time $T$, when requests are assigned according to the random assignment base policy, $\hat \pi$, is given by,
\begin{equation}
{Z_{T}^{\hat{\pi}}} = \sum\limits_{t = 1}^T {\sum\limits_{q = {R_{t - 1}} + 1}^{{R_{t}}} {d{'}\left( {v_{r_{q}}^{\hat{\pi}},{\rho _{{r_q}}}} \right) + d{'}\left( {{\rho _{{r_q}}},{\delta _{{r_q}}}} \right)} } 
\end{equation}
where $v_{r_{q}}^{\hat{\pi}}$ corresponds to the location of the agent that got assigned to request $r_q$ based on the random assignment policy $\hat{\pi}$.
We now analyze the expected distance the fleet travels to service all requests according to the random assignment base policy $\hat \pi$, $E\left[ {Z_{T}^{\hat{\pi}}}\right]$, given by,
\begin{equation}
\begin{array}{l}
E\left[ Z_{T}^{\hat{\pi}} \right] \mathop =\limits^{(1)} E\left[ {\sum\limits_{t = 1}^T {\sum\limits_{q = {R_{t - 1}} + 1}^{{R_{t}}} {d{'}\left( {v_{r_{q}}^{\hat{\pi}},{\rho _{{r_q}}}} \right) + d{'}\left( {{\rho _{{r_q}}},{\delta _{{r_q}}}} \right)} } } \right]\\  
\mathop =\limits^{(2)} \sum\limits_{t = 1}^T {E\left[ {\sum\limits_{q = {R_{t - 1}} + 1}^{{R_{t}}} {d{'}\left( {v_{r_{q}}^{\hat{\pi}},{\rho _{{r_q}}}} \right) + d{'}\left( {{\rho _{{r_q}}},{\delta _{{r_q}}}} \right)} } \right]} \\ 
\mathop \ge\limits^{(3)}   \sum\limits_{t = 1}^T {\mathop {\min }\limits_{{v_{R_{t-1} + 1}},...,{v_{R_t}}} E\left[ {\sum\limits_{q = {R_{t - 1}} + 1}^{{R_{t}}} {d{'}\left( {v_{r_{q}},{\rho _{{r_q}}}} \right) + d{'}\left( {{\rho _{{r_q}}},{\delta _{{r_q}}}} \right)} } \right]} \\ 
\mathop \ge\limits^{(4)} E\left[ {\sum\limits_{t = 1}^T {\mathop {\min }\limits_{{v_{R_{t-1} + 1}},...,{v_{R_t}}} } \sum\limits_{q = {R_{t - 1}} + 1}^{{R_{t}}} {d{'}\left( {v_{r_{q}},{\rho _{{r_q}}}} \right) + d{'}\left( {{\rho _{{r_q}}},{\delta _{{r_q}}}} \right)} } \right]   \\
= E\left[ Z_{T}^{\bar{\pi}} \right] 
\end{array}  
\end{equation}
In $(1)$ the expectation is applied to the definition of ${Z_{T}^{\hat{\pi}}}$, and in $(2)$ we use the linearity of expectations. In $(3)$ we choose the assignment of agents to requests that minimizes the expected distance traveled at each time step by the fleet. To prove $(4)$ we use the argument presented in Appendix~\ref{appendix:min_inequality}. The last equality follows from the expression for ${Z_{T}^{\bar{\pi}}}$ given in equation \ref{e120}. Therefore we establish that application of the instantaneous assignment base policy results in a smaller expected distance to service all requests that enter the system compared to the random assignment policy. Hence, choosing the number of cooperative agents according to the bound obtained for the random assignment policy in Theorem~\ref{theorem:new_sufficient_condition_RA} ensures stability for the instantaneous assignment base policy as well. 
\end{proof}

\section{Case Study and Empirical Results: Autonomous Taxicab Routing in San Francisco}
In this section we present empirical experiments performed to validate our theoretical claims (see Sec.~\ref{sec:theoretical_results}). We consider a case study on transportation that uses real ride requests for San Francisco's taxi service \cite{piorkowski2009crawdad}.

We first validate the theoretical results presented in Theorem~\ref{theorem:instability} by executing the random assignment policy for a large time horizon and plotting the number of outstanding requests at each time step. Our empirical results show that once the proportion of adversaries surpasses the threshold given in Theorem~\ref{theorem:instability}, the random assignment policy $\hat{\pi}$ becomes unstable accumulating outstanding requests over time. In addition to this, we also empirically show that a similar proportion of adversarial agents will result in instability for the instantaneous assignment policy $\bar{\pi}$. We include this empirical result to motivate potential future work on the characterization of the proportion of adversarial agents that will make instantaneous assignment unstable.

We then validate the results presented in Theorem~\ref{theorem:new_sufficient_condition_RA} and Theorem~\ref{theorem:new_sufficient_condition_IA} by executing both the instantaneous assignment and the random assignment policies. We consider two fleet sizes, the fleet size that satisfies the sufficient condition for a fully cooperative fleet \cite{garces2024approximate}, and our new sufficiently large fleet size. 
Our empirical results show that if we choose the fleet size according to the sufficient condition on Theorem~\ref{theorem:new_sufficient_condition_RA}, both instantaneous assignment and random assignment can recover stability, having a uniformly bounded number of outstanding requests over time. It is important to note that the increase in the number of agents in the fleet is smaller than the number of adversaries in the original fleet, showing that we can still recover stability by inserting fewer agents in the system than the number that would be required to replace all adversaries.

\subsection{Implementation Details}

\begin{wrapfigure}[12]{l}{0.42\linewidth}
\vspace{-10pt}
\centering
\includegraphics[width=0.99\linewidth]{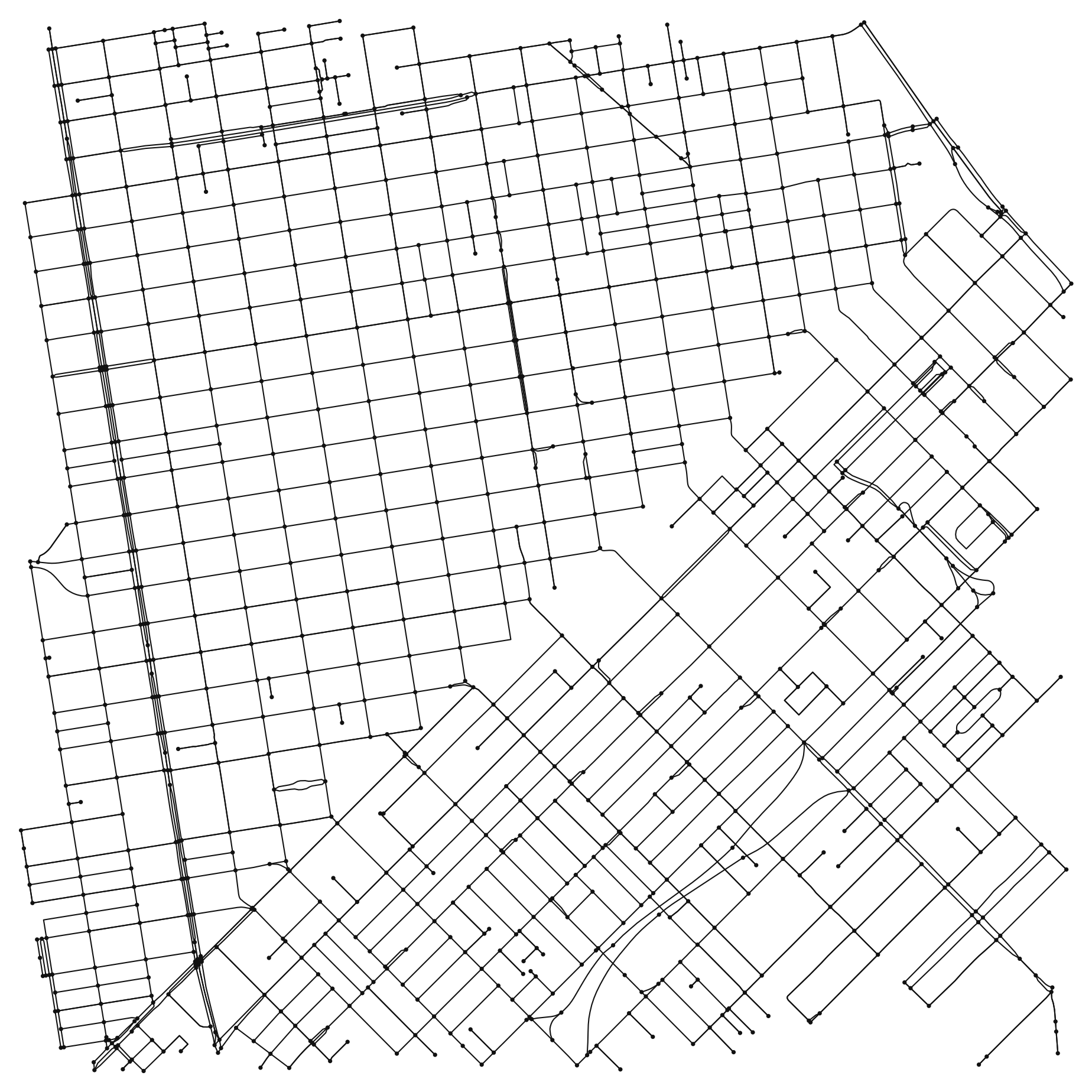} 
\vspace{-20pt}
\caption{\small{Map of the environment used for numerical simulations.}}
\label{fig:sf_map}
\end{wrapfigure}

For our simulation environment, we consider a section of San Francisco with a radius of $1500$ meters centered around the financial district (see Fig.~\ref{fig:sf_map}). This region has $1026$ intersections and $2300$ directed streets. 

We assume that each time step in the system corresponds to $1$ minute. We choose $\Delta=15$, which corresponds to adversarial agents having a maximum deviation of $15$ minutes. For our simulation, we assume that all adversarial agents execute actions that will achieve the maximum deviation $\Delta$ for both the route executed for picking up the assigned request and the route executed to drop-off the assigned request. Since taking time to infinity is not possible in the simulation, we consider a large finite horizon of $720$ time steps, which corresponds to $720$ minutes ($12$ hours). For the experiments where the maximum proportion of adversarial agents is known, we set $F_{\text{max}} = F$, and assume that the proportion of adversarial agents $F$ stays fixed over the entire time horizon $T$.

All experiments were executed in an AMD Threadripper Pro WRX80. Policies are evaluated on $100$ different instantiations of the random variables associated with the requests, and results are reported as averages over these $100$ runs.

\subsection{Estimating Probability Distributions}
\label{subsec:estimating_prob_dist}

Similarly to \cite{garces2024approximate}, we estimate the requests probability distributions using historical request data. More specifically, we obtain estimates of $p_\rho$ and $p_\delta$ using the relative frequency of historical requests that were picked up and dropped off within the section of San Francisco chosen for the experiment. We obtain an estimate for the probability distribution $p_\eta$ by considering the relative frequency of the number of requests that enter the system at each minute.

We set $p_\xi$ the distribution of the initial locations of the agents to be $p_{\delta}$, in order to model the possibility of agents already being deployed according to the distribution of requests' drop-offs. This would be the case if we consider that the current time horizon is starting after the end of a previous time horizon, while the agents are still circulating the streets.

Using these assumptions and the estimation procedure described above, we get that that $E[\eta] \approx 1.02$, $E[d(\xi, \rho)] \approx 17.47$, $E[d(v_{\text{rand}}, \rho)] \approx 17.62$, $E[d(\rho, \delta)] \approx 16.27$, and $\mathit{WD}(\delta, \rho) \approx 1.09$, where we approximate the Wasserstein distance using the procedure described in \cite{spieser2014}. 

\subsection{Proportion of Adversarial Agents That Provably Results in Instability}
In this subsection, we look at the effect of adversarial agents on stability. Using the expected values presented in Sec.~\ref{subsec:estimating_prob_dist}, we get that the sufficiently large fleet size needed for a fully cooperative fleet of agents (see Theorem~\ref{theorem:coop_sufficient_condition}) is $N' \geq 35$. Using the result from Theorem~\ref{theorem:instability}, we also get that the proportion of adversarial agents that will make random assignment with such a fleet size unstable is $F > 1/2$. 

\begin{figure}
    \centering
    \vspace{5pt}
    \includegraphics[width=0.99\linewidth]{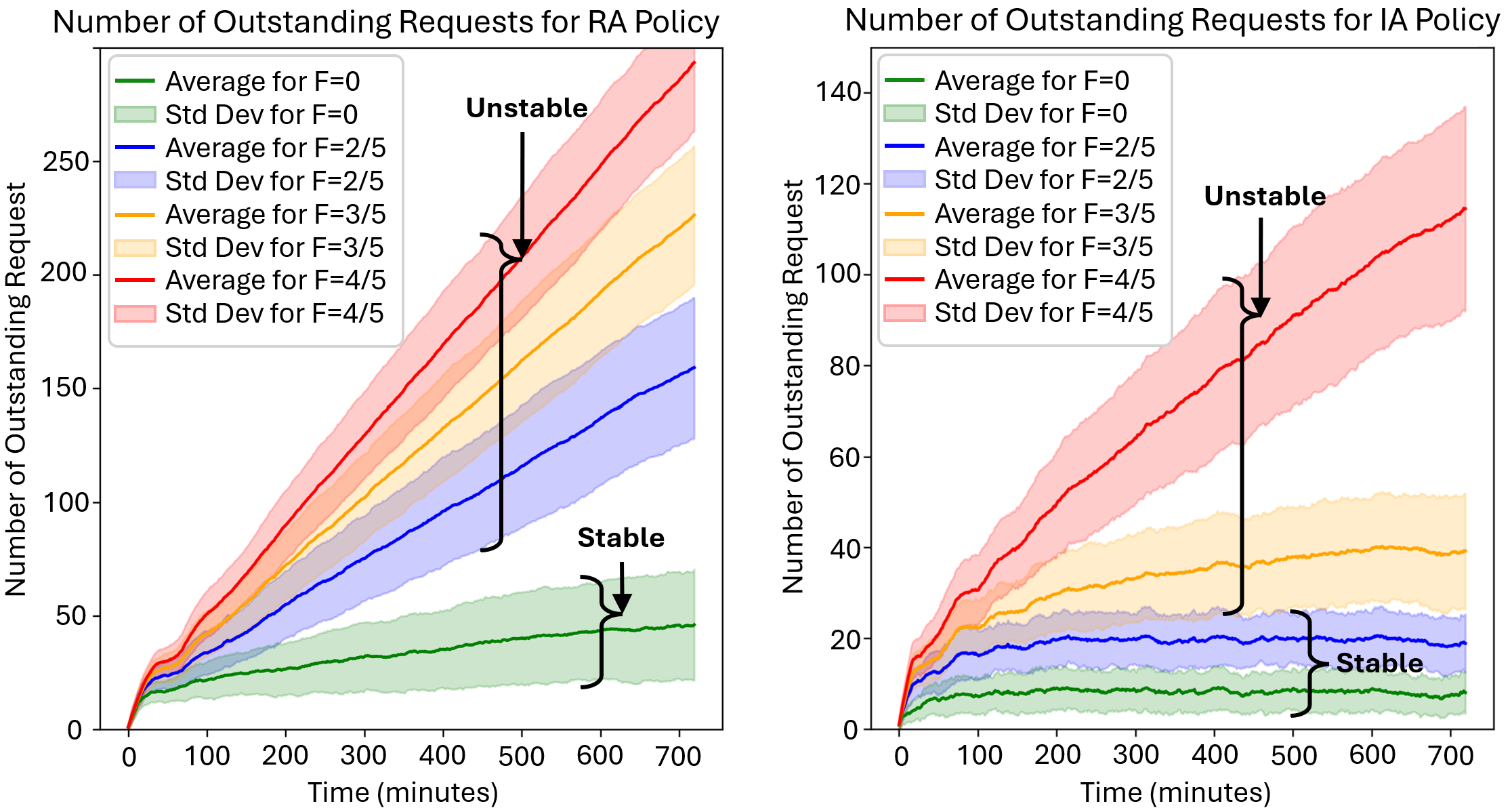}
    \vspace{-10pt}
    \caption{Number of outstanding requests for random assignment (RA) and Instantaneous Assignment (IA) policies with a fleet size of $35$ agents and different proportions of adversarial agents $F$.}
    \label{ref:instability_results}
\end{figure}

In Fig.~\ref{ref:instability_results} (left), we empirically show that once the proportion of adversarial agents surpasses the sufficient condition $F > 1/2$, the random assignment policy $\hat{\pi}$ is unstable for a fleet size of $35$ agents. From the results, we also get that even for a proportion of adversarial agents of $2/5$, the random assignment policy is still unstable. This happens due to the fact that the necessary condition (Corollary~\ref{coro:coop_necessary_condition}) used to derive the proportion of adversarial agents $F$ is derived using a loose lower bound. The lower bound was originally designed for instantaneous assignment with reassignment, and it is being applied to random assignment. Random assignment as a routing policy results in assignments with significantly higher travel times than the assignments obtained from instantaneous assignment with reassignment. We also include the plot of the number of outstanding requests as a function of time for the random policy with no adversarial agents (i.e. $F=0$) for completeness. This line starts growing at earlier time steps, but eventually plateaus, showing that the policy eventually becomes stable.

We also present results on the instantaneous assignment policy $\hat{\pi}$ in Fig.~\ref{ref:instability_results} (right). As shown in the figure, once the proportion of adversarial agents surpasses the sufficient condition $F > 1/2$, the instantaneous assignment policy becomes unstable. In the figure, lines for both $F=3/5$ and $F=4/5$ grow unboundedly. However, the line for $F=4/5$ grows faster. For a zoomed in depiction of these lines, please refer to Fig.~\ref{ref:ia_stability_results}. We use these results as motivation for future work, where we will try to characterize a sufficient condition on $F$ for the instantaneous assignment policy.

After having verified that adversarial agents induce instability, we will now verify that the new sufficient condition on the fleet size when we assume that we know the proportion of adversarial agents presented in Theorems~\ref{theorem:new_sufficient_condition_RA} and \ref{theorem:new_sufficient_condition_IA} maintains the number of outstanding requests uniformly bounded over time in practice, guaranteeing stability.

\subsection{New Sufficient Condition on the Fleet Size for Stability}
In this subsection, we look at the new sufficient condition for stability introduced in Theorems~\ref{theorem:new_sufficient_condition_RA}, and \ref{theorem:new_sufficient_condition_IA}. Using the expected values presented in Sec.~\ref{subsec:estimating_prob_dist}, we get that the sufficiently large fleet for both instantaneous assignment and random assignment when $F=2/5$ is $N\geq 47$, when $F=3/5$ is $N \geq 53$, and when $F=4/5$ is $N \geq 59$. 

\begin{figure*}[ht]
    \centering
    \vspace{5pt}
    \includegraphics[width=0.73\textwidth]{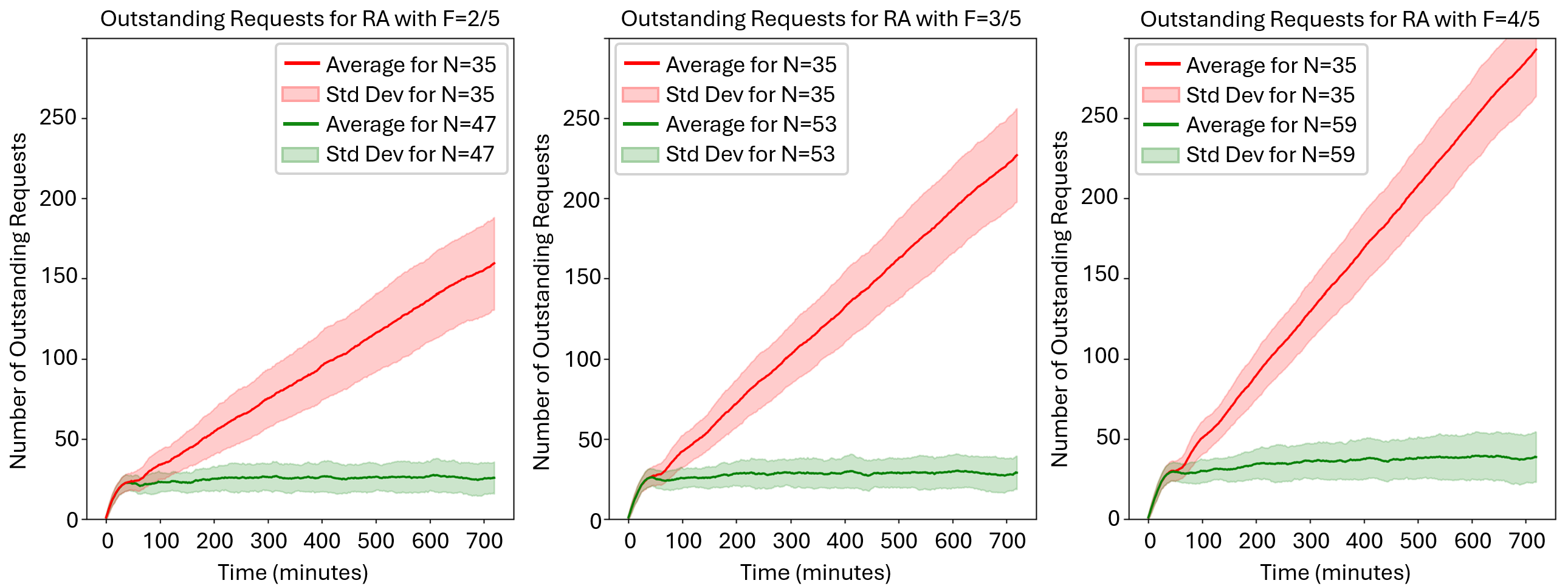}
    \vspace{-10pt}
    \caption{Number of outstanding requests for Random Assignment (RA) policy with different proportions of adversarial agents $F$ and different fleet sizes $N$.}
    \label{ref:rand_stability_results}
\end{figure*}

In Fig.~\ref{ref:rand_stability_results} we present the number of outstanding requests for the random assignment policy $\hat{\pi}$ for the three sufficiently large fleet sizes calculated above. We also include results for $N=35$ for all three proportions of adversarial agents to show that the sufficiently large fleet size for an all cooperative fleet is unstable for the proportions of adversarial agents displayed. The results in Fig.~\ref{ref:rand_stability_results} show that by increasing the fleet size by $2\Delta E\left[ \eta  \right] {F}$, we can provably recover stability for the random assignment policy $\hat{\pi}$ as predicted by Theorem~\ref{theorem:new_sufficient_condition_RA}.

\begin{figure*}[ht]
    \centering
    \includegraphics[width=0.73\textwidth]{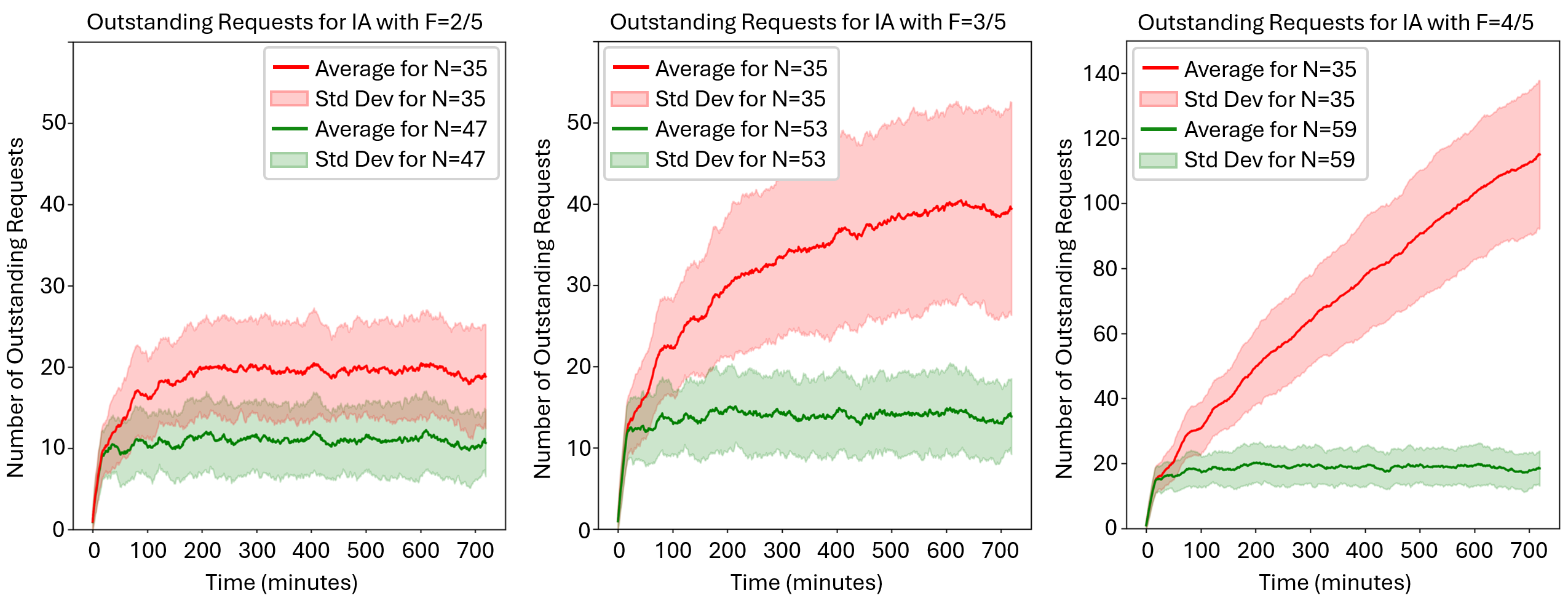}
    \vspace{-10pt}
    \caption{Number of outstanding requests for instantaneous assignment (IA) policy with different proportions of adversarial agents $F$ and different fleet sizes $N$. Note that the vertical axis scale is different for all three sub-figures, and the vertical axis scale is smaller than the one for Fig.~\ref{ref:rand_stability_results}}
    \label{ref:ia_stability_results}
\end{figure*}

Similarly, in Fig.~\ref{ref:ia_stability_results}  we present the number of outstanding requests for the instantaneous assignment policy for the three sufficiently large fleet sizes calculated above. We also include results for $N=35$ for all three proportions of adversarial agents to show that the sufficiently large fleet size for an all cooperative fleet is unstable for the $F = 3/5$, and $F=4/5$ cases. The results depicted in the figure show that by increasing the fleet size by $2\Delta E\left[ \eta  \right] {F}$, we can provably recover stability for the instantaneous assignment policy $\bar{\pi}$ as predicted by Theorem~\ref{theorem:new_sufficient_condition_IA}.

In a more intuitive sense, these results for the random assignment policy and the instantaneous assignment policy show that we can recover stability by increasing the fleet size based on the proportion of adversarial agents and the maximum deviation. If the maximum deviation $2\Delta E[\eta]$ is smaller than the original fleet size, then the increase in the number of cooperative agents will be less than the number of adversarial agents in the original fleet. 
Let's look at the specific scenario considered in these numerical experiments. 
When $F=2/5$ and the original fleet size is $N=35$ agents, we get that there are $14$ adversarial agents and $21$ cooperative agents. Since we maintain $F$ when calculating the new sufficiently large fleet size, we obtain that the number of adversarial agents increases to $19$, while the number of cooperative agents increases to $28$ agents. This shows that even though there were $14$ adversarial agents in the original fleet, we only need to increase the number of cooperative agents by $7$ agents to recover stability as expected from our theoretical results. This same observation also holds for the cases when $F=3/5$ and $F=4/5$, showing that our theoretical bounds on the new sufficiently large fleet size result in an increase in the number of cooperative agents that is smaller than the number of adversarial agents in the original fleet size.

\section{Conclusion}
This work addresses autonomous routing problems where adversarial agents impose bounded delays in servicing pickup-and-delivery requests. These delays can accumulate over time, causing an overload of unfulfilled requests and ultimately destabilizing the routing policy. At first, we establish a theoretical threshold on the proportion of adversarial agents that if exceeded leads to the instability of the random assignment policy. Furthermore, we introduce a new sufficient condition on the fleet size that ensures the stability of the random assignment policy and the instantaneous assignment policy in the presence of bounded-delay adversarial agents. We validate our theoretical results through a case study on autonomous taxi routing, leveraging real transportation request data from San Francisco taxicabs. Even though our case study only considers taxicab routing, our results are applicable to other routing settings, including warehouse package routing, freight management, and robotic delivery. As future work, we want to consider other attack models, including attack models where agents actively try to sabotage the system, and models where agents collude to maximize their impacts. We also want to develop monitoring systems to detect and mitigate the effects of adversaries under these new attack models.

\section{Acknowledgements}
We thank Orhan Eren Akgün and Áron Vékássy for their useful discussions of the theoretical results. We gratefully acknowledge AFOSR award $\#$FA9550-22-1-0223 and DARPA YFA award $\#$D24AP00319-00 for partial support of this work.

\appendices
\section{}
\label{appendix:fa_bound}

In this appendix we seek an upper bound on the expected proportion of adversarial agents in the fleet at each assignment index $q$. To achieve this, we consider a worst-case scenario where cooperative agents experience the same delays as adversarial agents. By using this worst-case scenario, we establish symmetry in the assignment dynamics for both adversarial and cooperative agents. In practice, cooperative agents complete requests faster, making them available for reassignment sooner. As a result, they are assigned more tasks on average, while adversarial agents spend more time servicing requests and are assigned fewer tasks in expectation. Consequently, the expected proportion of adversarial agents available to service requests is lower than their overall fleet proportion, $F$, implying that $E\left[ F_t \right] \le F$, while the expected proportion of cooperative agents available to service requests is higher than their overall fleet proportion $1- F$, implying that $E\left[ 1 - F_t \right] \ge 1- F$.

We formally prove this proposition. Denote by $x_t$ the number of adversarial agents available at time $t$ and by $y_t$ the total number of agents at time $t$.

From the algorithm's perspective, all agents are equivalent. Therefore, for any agent $i$, the probability of being available is the same by symmetry.  Now consider $E\left[ {\frac{{{x_t}}}{{{y_t}}}} \right]$, which gives us the expectation of assigning a request to an adversarial agent at time $t$. 
\begin{equation}
\begin{array}{l}
E\left[ {\frac{{{x_t}}}{{{y_t}}}} \right] \\ \mathop =\limits^{(1)}  \sum\limits_{k = 1}^N {E\left[ {\left. {\frac{{{x_t}}}{{{y_t}}}} \right|{y_t} = k} \right]p\left( {{y_t} = k} \right)} \\ \mathop =\limits^{(2)}  \sum\limits_{k = 1}^N {\frac{{{F}k}}{k}p\left( {{y_t} = k} \right)} \\\mathop =\limits^{(3)} {F}\sum\limits_{k = 1}^N {p\left( {{y_t} = k} \right)} \\ \mathop =\limits^{(4)} {F}
\end{array}
\label{e400}   
\end{equation}
In $(1)$ we use the law of total probability. By symmetry, any agent has equal likelihood of being in the selected subset $k$. For a randomly selected subset, the probability of being adversarial is $Fk$ implying $(2)$. $(3)$ results from cancellation of terms and in $(4)$ we sum over all possible options for the fleet size implying that $\sum\limits_{k = 1}^N {p\left( {{y_t} = k} \right)}=1$. Following similar arguments, we obtain that the expectation of assigning a request to a cooperative agent at time $t$ is $E\left[ {\frac{{{y_t-x_t}}}{{{y_t}}}} \right]=1-F $.

\section{}
\label{appendix:min_inequality}
In this appendix we prove that,
\begin{equation}
\begin{array}{l}
\sum\limits_{t = 1}^T {\mathop {\min }\limits_{{v_1},...,{v_{{n_a} + {n_c}}}} E\left[ {\sum\limits_{q = {R_{t - 1}} + 1}^{{R_{t}}} {d{'_q}\left( {{v_{{q}}},{\rho _{{r_q}}}} \right) + d{'_q}\left( {{\rho _{{r_q}}},{\delta _{{r_q}}}} \right)} } \right]}   \\ \ge E\left[ {\sum\limits_{t = 1}^T {\mathop {\min }\limits_{{v_1},...,{v_{{n_a} + {n_c}}}} } \sum\limits_{q = {R_{t - 1}} + 1}^{{R_{t}}} {d{'_q}\left( {{v_{{q}}},{\rho _{{r_q}}}} \right) + d{'_q}\left( {{\rho _{{r_q}}},{\delta _{{r_q}}}} \right)} } \right]
\end{array}    
\label{e130}
\end{equation}
Denote by $Q\left( {\vec v } \right)$ the total distance required to service all requests, by a fleet of agents at locations $\vec v$.   
\begin{equation}
\begin{array}{l}
Q\left( {\vec v } \right) = \sum\limits_{i = 1}^n {d\left( {v_i,{\rho _i}} \right) + d\left( {{\rho _i},{\delta _i}} \right)}    
\end{array}
\label{e131}
\end{equation}
Where $v_i$ corresponds to the location of agent $\ell_i$.
Choosing the assignment of agents that minimizes $Q\left( {\vec v } \right)$ clearly satisfies, 
\begin{equation}
\begin{array}{l}
Q\left( {\vec v } \right) \ge \mathop {\min }\limits_{\vec v '} Q\left( {\vec v '} \right) 
\end{array}
\label{e132}
\end{equation}
Applying the expectation operator to both sides of the inequality yields,
\begin{equation}
\begin{array}{l}
E\left[ {Q\left( {\vec v } \right)} \right] \ge E\left[ {\mathop {\min }\limits_{\vec v '} Q\left( {\vec v '} \right)} \right]   
\end{array}
\label{e133}
\end{equation}
Taking the minimum over the assignment of agents to requests on both sides of equation \ref{e133} yields,
\begin{equation}
\begin{array}{l}
\mathop {\min }\limits_{\vec v ''} E\left[ {Q\left( {\vec v ''} \right)} \right] \ge \mathop {\min }\limits_{\vec v ''} E\left[ {\mathop {\min }\limits_{\vec v '} Q\left( {\vec v '} \right)} \right] \mathop =\limits^{(5)} E\left[ {\mathop {\min }\limits_{\vec v '} Q\left( {\vec v '} \right)} \right]   
\end{array}
\label{e134}
\end{equation}
Where $(5)$ is satisfied since $E\left[ {\mathop {\min }\limits_{\vec v '} Q\left( {\vec v '} \right)} \right]$ is a constant that does not depend on ${\vec v ''}$. Hence we obtain the desired result,
\begin{equation}
\begin{array}{l}
\mathop {\min }\limits_{\vec v ''} E\left[ {Q\left( {\vec v ''} \right)} \right] \ge E\left[ {\mathop {\min }\limits_{\vec v '} Q\left( {\vec v '} \right)} \right]   
\end{array}
\label{e135}
\end{equation}

\bibliographystyle{IEEEtran}
\bibliography{refs}

\begin{thebibliography}{10}
\providecommand{\url}[1]{#1}
\csname url@samestyle\endcsname
\providecommand{\newblock}{\relax}
\providecommand{\bibinfo}[2]{#2}
\providecommand{\BIBentrySTDinterwordspacing}{\spaceskip=0pt\relax}
\providecommand{\BIBentryALTinterwordstretchfactor}{4}
\providecommand{\BIBentryALTinterwordspacing}{\spaceskip=\fontdimen2\font plus
\BIBentryALTinterwordstretchfactor\fontdimen3\font minus \fontdimen4\font\relax}
\providecommand{\BIBforeignlanguage}[2]{{%
\expandafter\ifx\csname l@#1\endcsname\relax
\typeout{** WARNING: IEEEtran.bst: No hyphenation pattern has been}%
\typeout{** loaded for the language `#1'. Using the pattern for}%
\typeout{** the default language instead.}%
\else
\language=\csname l@#1\endcsname
\fi
#2}}
\providecommand{\BIBdecl}{\relax}
\BIBdecl

\bibitem{Alonso2017}
\BIBentryALTinterwordspacing
J.~Alonso-Mora, S.~Samaranayake, A.~Wallar, E.~Frazzoli, and D.~Rus, ``On-demand high-capacity ride-sharing via dynamic trip-vehicle assignment,'' \emph{Proceedings of the National Academy of Sciences}, vol. 114, no.~3, pp. 462--467, 2017. [Online]. Available: \url{https://www.pnas.org/doi/abs/10.1073/pnas.1611675114}
\BIBentrySTDinterwordspacing

\bibitem{garces2023multiagent}
D.~Garces, S.~Bhattacharya, S.~Gil, and D.~Bertsekas, ``Multiagent reinforcement learning for autonomous routing and pickup problem with adaptation to variable demand,'' in \emph{2023 IEEE International Conference on Robotics and Automation (ICRA)}.\hskip 1em plus 0.5em minus 0.4em\relax IEEE, 2023, pp. 3524--3531.

\bibitem{garces2024approximate}
D.~Garces, S.~Bhattacharya, D.~Bertsekas, and S.~Gil, ``Approximate multiagent reinforcement learning for on-demand urban mobility problem on a large map,'' in \emph{2024 IEEE International Conference on Robotics and Automation (ICRA)}.\hskip 1em plus 0.5em minus 0.4em\relax IEEE, 2024, pp. 6843--6849.

\bibitem{Chen2021}
Z.~Chen, J.~Alonso-Mora, X.~Bai, D.~D. Harabor, and P.~J. Stuckey, ``Integrated task assignment and path planning for capacitated multi-agent pickup and delivery,'' \emph{IEEE Robotics and Automation Letters}, vol.~6, no.~3, pp. 5816--5823, 2021.

\bibitem{Wilde2024}
N.~Wilde and J.~Alonso-Mora, ``Statistically distinct plans for multiobjective task assignment,'' \emph{IEEE Transactions on Robotics}, vol.~40, pp. 2217--2232, 2024.

\bibitem{mari2025online}
M.~Mari, M.~Paw{\l}owski, R.~Ren, and P.~Sankowski, ``Online matching with delays and stochastic arrival times,'' \emph{Theory of Computing Systems}, vol.~69, no.~1, pp. 1--46, 2025.

\bibitem{smith2008dynamic}
S.~L. Smith, M.~Pavone, F.~Bullo, and E.~Frazzoli, ``Dynamic vehicle routing with heterogeneous demands,'' in \emph{2008 47th IEEE Conference on Decision and Control}.\hskip 1em plus 0.5em minus 0.4em\relax IEEE, 2008, pp. 1206--1211.

\bibitem{zhang2016control}
R.~Zhang and M.~Pavone, ``Control of robotic mobility-on-demand systems: a queueing-theoretical perspective,'' \emph{The International Journal of Robotics Research}, vol.~35, no. 1-3, pp. 186--203, 2016.

\bibitem{spieser2014}
K.~Spieser, K.~Treleaven, R.~Zhang, E.~Frazzoli, D.~Morton, and M.~Pavone, ``Toward a systematic approach to the design and evaluation of automated mobility-on-demand systems: A case study in singapore,'' \emph{Road Vehicle Automation. Lecture Notes on Mobility}, pp. 229--245, 04 2014.

\bibitem{yu2022delay}
Q.~Yu, Y.~Zhang, and Y.-P. Zhou, ``Delay information in virtual queues: A large-scale field experiment on a major ride-sharing platform,'' \emph{Management Science}, vol.~68, no.~8, pp. 5745--5757, 2022.

\bibitem{Yuan2016}
C.~Yuan, J.~Thai, and A.~M. Bayen, ``Zubers against zlyfts apocalypse: An analysis framework for dos attacks on mobility-as-a-service systems,'' in \emph{2016 ACM/IEEE 7th International Conference on Cyber-Physical Systems (ICCPS)}, 2016, pp. 1--10.

\bibitem{Thai2018}
J.~Thai, C.~Yuan, and A.~M. Bayen, ``Resiliency of mobility-as-a-service systems to denial-of-service attacks,'' \emph{IEEE Transactions on Control of Network Systems}, vol.~5, no.~1, pp. 370--382, 2018.

\bibitem{Kearney2018induction}
G.~Kearney and M.~Fardad, ``On the induction of cascading failures in transportation networks,'' in \emph{2018 IEEE Conference on Decision and Control (CDC)}.\hskip 1em plus 0.5em minus 0.4em\relax IEEE, 2018, pp. 1821--1826.

\bibitem{shishika2018local}
D.~Shishika and V.~Kumar, ``Local-game decomposition for multiplayer perimeter-defense problem,'' in \emph{2018 IEEE Conference on Decision and Control (CDC)}.\hskip 1em plus 0.5em minus 0.4em\relax IEEE, 2018, pp. 2093--2100.

\bibitem{shishika2019team}
D.~Shishika, J.~Paulos, M.~R. Dorothy, M.~A. Hsieh, and V.~Kumar, ``Team composition for perimeter defense with patrollers and defenders,'' in \emph{2019 IEEE 58th Conference on Decision and Control (CDC)}.\hskip 1em plus 0.5em minus 0.4em\relax IEEE, 2019, pp. 7325--7332.

\bibitem{kailkhura2014asymptotic}
B.~Kailkhura, Y.~S. Han, S.~Brahma, and P.~K. Varshney, ``Asymptotic analysis of distributed bayesian detection with byzantine data,'' \emph{IEEE Signal Processing Letters}, vol.~22, no.~5, pp. 608--612, 2014.

\bibitem{Altman2010adversarial}
E.~Altman, T.~Ba{\c{s}}ar, and V.~Kavitha, ``Adversarial control in a delay tolerant network,'' in \emph{International Conference on Decision and Game Theory for Security}.\hskip 1em plus 0.5em minus 0.4em\relax Springer, 2010, pp. 87--106.

\bibitem{Zhu2012deceptive}
Q.~Zhu, A.~Clark, R.~Poovendran, and T.~Ba{\c{s}}ar, ``Deceptive routing games,'' in \emph{2012 IEEE 51st IEEE Conference on Decision and Control (CDC)}.\hskip 1em plus 0.5em minus 0.4em\relax IEEE, 2012, pp. 2704--2711.

\bibitem{Chu2023}
K.-F. Chu and W.~Guo, ``Passenger spoofing attack for artificial intelligence-based mobility-as-a-service,'' in \emph{2023 IEEE 26th International Conference on Intelligent Transportation Systems (ITSC)}, 2023, pp. 4874--4880.

\bibitem{Chu2024}
K.~F. Chu and W.~Guo, ``Multi-agent reinforcement learning-based passenger spoofing attack on mobility-as-a-service,'' \emph{IEEE Transactions on Dependable and Secure Computing}, vol.~21, no.~6, pp. 5565--5581, 2024.

\bibitem{Sharma2019attacks}
P.~Sharma, D.~Austin, and H.~Liu, ``Attacks on machine learning: Adversarial examples in connected and autonomous vehicles,'' in \emph{2019 IEEE International Symposium on Technologies for Homeland Security (HST)}.\hskip 1em plus 0.5em minus 0.4em\relax IEEE, 2019, pp. 1--7.

\bibitem{Qayyum2020securing}
A.~Qayyum, M.~Usama, J.~Qadir, and A.~Al-Fuqaha, ``Securing connected \& autonomous vehicles: Challenges posed by adversarial machine learning and the way forward,'' \emph{IEEE Communications Surveys \& Tutorials}, vol.~22, no.~2, pp. 998--1026, 2020.

\bibitem{Guo2023}
S.~Guo, H.~Chen, M.~Rahman, and X.~Qian, ``Dca: Delayed charging attack on the electric shared mobility system,'' \emph{IEEE Transactions on Intelligent Transportation Systems}, vol.~24, no.~11, pp. 12\,793--12\,805, 2023.

\bibitem{Poudel2024}
S.~Poudel, M.~Abouyoussef, J.~E. Baugh, and M.~Ismail, ``Attack design for maximum malware spread through evs commute and charge in power-transportation systems,'' \emph{IEEE Systems Journal}, vol.~18, no.~3, pp. 1809--1820, 2024.

\bibitem{cavorsi2022adaptive}
M.~Cavorsi, N.~Jadhav, D.~Salda{\~n}a, and S.~Gil, ``Adaptive malicious robot detection in dynamic topologies,'' in \emph{2022 IEEE 61st Conference on Decision and Control (CDC)}.\hskip 1em plus 0.5em minus 0.4em\relax IEEE, 2022, pp. 2236--2243.

\bibitem{yu2024sensing}
Y.~Yu, A.~J. Thorpe, J.~Milzman, D.~Fridovich-Keil, and U.~Topcu, ``Sensing resource allocation against data-poisoning attacks in traffic routing,'' 2024.

\bibitem{francos2021search}
R.~M. Francos and A.~M. Bruckstein, ``Search for smart evaders with swarms of sweeping agents,'' \emph{IEEE Transactions on Robotics}, vol.~38, no.~2, pp. 1080--1100, 2021.

\bibitem{Rivas2019}
D.~Rivas, J.~Jim{\'e}nez-Jan{\'e}, and L.~Ribas-Xirgo, ``Auction model for transport order assignment in agv systems,'' in \emph{Advances in Physical Agents}, R.~Fuentetaja~Piz{\'a}n, {\'A}.~Garc{\'i}a~Olaya, M.~P. Sesmero~Lorente, J.~A. Iglesias~Mart{\'i}nez, and A.~Ledezma~Espino, Eds.\hskip 1em plus 0.5em minus 0.4em\relax Cham: Springer International Publishing, 2019, pp. 227--241.

\bibitem{piorkowski2009crawdad}
M.~Piorkowski, N.~Sarafijanovic-Djukic, and M.~Grossglauser, ``Crawdad data set epfl/mobility (v. 2009-02-24),'' 2009.

\bibitem{bertsekas1988auction}
D.~Bertsekas, ``The auction algorithm: A distributed relaxation method for the assignment problem,'' \emph{Annals of operations research}, vol.~14, no.~1, pp. 105--123, 1988.

\bibitem{bertsekas2024new}
D.~P. Bertsekas, ``New auction algorithms for the assignment problem and extensions,'' \emph{Results in control and optimization}, vol.~14, p. 100383, 2024.

\bibitem{crouse2016implementing}
D.~F. Crouse, ``On implementing 2d rectangular assignment algorithms,'' \emph{IEEE Transactions on Aerospace and Electronic Systems}, vol.~52, no.~4, pp. 1679--1696, 2016.

\bibitem{garces2024surge}
D.~Garces and S.~Gil, ``Surge routing: Event-informed multiagent reinforcement learning for autonomous rideshare,'' in \emph{Proceedings of the 23rd International Conference on Autonomous Agents and Multiagent Systems}, 2024, pp. 641--650.

\bibitem{ruschendorf1985wasserstein}
L.~R{\"u}schendorf, ``The wasserstein distance and approximation theorems,'' \emph{Probability Theory and Related Fields}, vol.~70, no.~1, pp. 117--129, 1985.

\bibitem{guerrero2017formations}
L.~Guerrero-Bonilla, A.~Prorok, and V.~Kumar, ``Formations for resilient robot teams,'' \emph{IEEE Robotics and Automation Letters}, vol.~2, no.~2, pp. 841--848, 2017.

\bibitem{sundaram2010distributed}
S.~Sundaram and C.~N. Hadjicostis, ``Distributed function calculation via linear iterative strategies in the presence of malicious agents,'' \emph{IEEE Transactions on Automatic Control}, vol.~56, no.~7, pp. 1495--1508, 2010.

\bibitem{pasqualetti2011consensus}
F.~Pasqualetti, A.~Bicchi, and F.~Bullo, ``Consensus computation in unreliable networks: A system theoretic approach,'' \emph{IEEE Transactions on Automatic Control}, vol.~57, no.~1, pp. 90--104, 2011.

\bibitem{saulnier2017resilient}
K.~Saulnier, D.~Saldana, A.~Prorok, G.~J. Pappas, and V.~Kumar, ``Resilient flocking for mobile robot teams,'' \emph{IEEE Robotics and Automation letters}, vol.~2, no.~2, pp. 1039--1046, 2017.

\bibitem{mitzenmacher2017probability}
M.~Mitzenmacher and E.~Upfal, \emph{Probability and computing: Randomization and probabilistic techniques in algorithms and data analysis}.\hskip 1em plus 0.5em minus 0.4em\relax Cambridge university press, 2017.

\end{thebibliography}

\end{document}